\documentclass[12pt,onecolumn,final]{IEEEtran} 
\usepackage{amsthm}

\usepackage{graphics}                
\usepackage{url}
\usepackage{bm,booktabs}
\usepackage{amsfonts,amssymb}
\usepackage[ruled]{algorithm2e}
\usepackage{mathrsfs}
\usepackage{algorithmic}
\usepackage{subcaption}
\usepackage{graphicx,epsfig}
\usepackage{amsmath}
\usepackage[all]{xy}
\usepackage{caption,fixltx2e}
\usepackage[flushleft]{threeparttable}  
\providecommand{\keyword}[1]{\textbf{Keywords: } #1}

\def \A   { \mathscr{A} }
\def \L   { \mathscr{L} }
\def \Rs  { \Re^{m \times n} }
\def \IRq {IRucL-$q$}
\def \rank { \mathrm{rank} }
\def \tr  { \mathrm{tr} }  
\def \trk { \mathrm{tr_k} }

\def \st  { \mathrm{s.t.} }
\def \vc  { \mathrm{vec} }
\def \Diag { \mathrm{Diag} }
\DeclareMathOperator*{\argmin}{ \mathrm{arg \, min} }
\def \R {{\mathcal R}}

\renewcommand{\theequation}{\arabic{section}.\arabic{equation}}

\newtheorem{theo}{Theorem}[section]

\newtheorem{lem}{Lemma}[section]

\renewcommand{\theequation}{\arabic{section}.\arabic{equation}}
\newcommand{\comments}[1]{}

\graphicspath{{pic/}}

\begin{document}
\title{Transformed Schatten-1 Iterative Thresholding Algorithms for Low Rank Matrix Completion}


\author{Shuai~Zhang,
~Penghang~Yin,
and~Jack~Xin
\thanks{
			The work was partially supported by NSF grant DMS-1222507 and DMS-1522383.
			They are with the Department of Mathematics, University of California, Irvine, CA, 92697, USA. \
			E-mail: szhang3@uci.edu, penghany@uci.edu, jxin@math.uci.edu. \
		    Phone: (949)-824-5309.  Fax: (949)-824-7993.  }}



\maketitle



\begin{abstract}
We study a non-convex low-rank promoting penalty function, 
the transformed Schatten-1 (TS1), and its applications in matrix completion. 
The TS1 penalty, as a matrix quasi-norm defined on its singular values, interpolates the  
rank and the nuclear norm through a nonnegative parameter $a \in (0, +\infty)$. 
We consider the unconstrained TS1 regularized low-rank matrix recovery problem and develop a fixed 
point representation for its global minimizer. The TS1 thresholding functions are in closed
analytical form for all parameter values. The TS1 threshold values differ in subcritical (supercritical) 
parameter regime where the TS1 threshold functions are continuous (discontinuous). We propose TS1 
iterative thresholding algorithms and compare them with some state-of-the-art algorithms on matrix
completion test problems.  For problems with known rank, a fully adaptive TS1 iterative 
thresholding algorithm consistently performs the best under different conditions, 
where ground truth matrices are generated by multivariate Gaussian, $(0,1)$ uniform and 
Chi-square distributions. 
For problems with unknown rank, TS1 algorithms with an additional rank estimation procedure      
approach the level of  \IRq \ which is an iterative reweighted algorithm, 
non-convex in nature and best in performance.
\end{abstract}

\keyword{
Transformed Schatten-1 penalty,  fixed point representation,  closed form thresholding function, iterative thresholding algorithms, 
 matrix completion.}


\section{Introduction} \label{section: intro}
Low rank matrix completion problems arise in many applications such as 
collaborative filtering in recommender systems \cite{CR_09,Jannach_12}, minimum order system and low-dimensional Euclidean embedding in control 
theory \cite{Fazel_01,Fazel_03}, network localization \cite{JiYe_13}, 
and others \cite{Fazel-siamreview}. 
The mathematical problem is:
\begin{equation}
\begin{array}{l}
  \min \limits_{X \in \Re^{m \times n}} \rank(X)  \ \ \ 
  \st \ \ X \in \L,
\end{array}
\end{equation}
where $ \L $ is a convex set. In this paper, we are interested 
in methods for solving the affine rank minimization problem (ARMP)
\begin{equation} \label{mini: rank and map}
\begin{array}{l}
  \min \limits_{X \in \Re^{m \times n}} \rank(X) \ \ \
  \st \ \ \A (X) = b \ \textit{in} \ \Re^{p},
\end{array}
\end{equation}
where linear transformation $\A: \Re^{m \times n} \rightarrow \Re^p$ 
and vector $b \in \Re^p$ are given. 
The matrix completion problem 
\begin{equation} \label{mini: rank and comple}
\begin{array}{l} 
  \min \limits_{X \in \Re^{m \times n}} \rank (X) \ \ \ 
  \st \ \ X_{i,j} = M_{i,j}, \ \ (i,j) \in \Omega
\end{array}
\end{equation}
is a special case of (\ref{mini: rank and map}), where $X$ and $M$ are both $m \times n$ 
matrices and $\Omega$ is a subset of index pairs $\{ (i,j) \}$. 
\medskip 

The optimization problems above are known to be NP-hard. 
Many alternative penalties have been utilized as proxies for finding 
low rank solutions in both the  
constrained and unconstrained settings: 
\begin{equation} \label{mini: general alternative constrain}
\begin{array}{l}
  \min \limits_{X \in \Re^{m \times n}} F(X) \ \ \ 
  \st \ \ \A (X) = b
\end{array}
\end{equation}
and 
\begin{equation} \label{mini: general regular}
  \min \limits_{X \in \Re^{m \times n}} { \frac{1}{2}\|\A(X) - b\|_2^2 + \lambda F(X) } .
\end{equation}

The penalty function $F(\cdot)$ is defined on singular values of matrix $X$, 
typically $F(X) = \sum \limits_{i} f(\sigma_i) $, where $\sigma_i$ is 
the $i$-th largest singular value of $X$ arranged in descending order. 
The Schatten $p$-norm (nuclear norm at $p=1$) results 
when $f(x) = x^p$, $p \in [0,1]$. 
At $p=0$ ($p=2$), $F$ is the rank (Frobenius norm). 
Recovering rank under suitable conditions for $p \in (0,1]$ 
has been extensively studied in theories and algorithms 
\cite{cai2010singular,candes2010power,CR_09,Kesh_10,
IRucLq,lu2014iterative,FPCA,Fazel,LMaFit}. 
Non-convex penalty based methods have shown better performance 
on hard problems \cite{IRucLq,Fazel}. 
There is also a novel method to solve the constrained problem 
(\ref{mini: general alternative constrain}), 
from the perspective of gauge dual \cite{friedlander2014gauge, friedlander2015low}.

Recently, a class of $\ell_1$ based non-convex penalty, the transformed 
$\ell_1$ (TL1), has been found effective and robust for compressed sensing 
problems \cite{Threshold-TL1,DCATL1}. 
TL1 interpolates $\ell_0$ and $\ell_1$, similar to $\ell_p$ quasi-norm 
($p \in (0,1)$). In the entire range of interpolation parameter, 
TL1 enjoys closed form iterative thresholding function, 
which is available for $\ell_p$ only at some specific values, 
like $p=0,1,1/2,2/3$, see \cite{hard-threshold-blumensath2012accelerated,
xian-2-3rds-cao2013fast-image,Daub_10,xian-half}.
This feature allows TL1 to perform fast and robust sparse minimization 
in a much wider range than $l_p$ quasi-norm. Moreover, the TL1 penalty boasts 
unbiasedness and Lipschitz continuity besides sparsity \cite{fan2001variable,transformed-l1}.  

It is the goal of this paper to extend TL1 penalty to TS1 (transformed Schatten-1)
for low rank matrix completion and compare it with state of the art methods 
in the literature. 

The rest of the paper is organized as follows. 
In section 2, we present the transformed Schatten-$1$ function (TS1),  
the TS1 regularized minimization problems, and a derivation of thresholding representation of the global minimum. 
In section 3, we propose two thresholding algorithms (TS1-s1 and TS1-s2) 
based on a fixed point equation of the global minimum.  
In section 4, we compare TS1 algorithms with some state-of-the-art 
algorithms through numerical experiments in low rank matrix recovery and 
image inpainting. Concluding remarks are in section 5.


\subsection{Notation}  \label{section: notation}
Here we set the notations for this paper. 
Two kinds of inner products are used in the following sections, 
one is between matrices and one is a bilinear operation for vectors: 
\begin{equation*}
\begin{array}{l}
(x,y) = \sum \limits_i  x_i y_i  \  \  \text{for vectors} \ x, y ; \\
\langle X, Y \rangle = \tr(Y^TX) = \sum \limits_{i,j} X_{i,j} Y_{i,j}  \  \  
\text{for matrices} \ X, Y.
\end{array}
\end{equation*}

Assume matrix $ X \in \Re^{m \times n}$ has $r$ positive singular values 
$\sigma_1 \geq \sigma_2 \geq ... \geq \sigma_r > 0$. Let us introduce 
some common matrix norms or quasi-norms as,  
\begin{itemize}
\item Nuclear norm:  $\|X\|_* = \sum \limits_{i = 1}^{ r } \sigma_i $;
\item Schatten $p$ quasi-norm: $\|X\|_p = ( \sum \limits_{i = 1}^{ r } \sigma_i^p )^{1/p} $, 
for $ p \in (0,1)$; 
\item Frobenius norm: $\|X\|_F = (\sum \limits_{i = 1}^{ r } \sigma_i^2 )^{\frac{1}{2}} $. Also 
$\|X\|_F^2 = \langle X, X \rangle = \sum \limits_{i,j} X_{i,j}^2$. 
\item Ky Fan $k$-norm: $\|X\|_{Fk} = \sum \limits_{i = 1}^{ k } \sigma_i $, for $1 \leq k \leq r$; 
\item Induced $L^2$ norm: $ \| X \|_{L^2} = \max \limits_{\|v\|_2 = 1} \|Xv\|_2 = \sigma_1$.
\end{itemize}
Define function $\vc(\cdot)$ to unfold one matrix columnwise into a vector.
So it is clearly that 
$ \|\vc(X)\|_2 = \|X\|_F$, 
where the left hand side norm is vector's $\ell_2$ norm.  

Define the shrinkage identity $k$ matrix $I_k^s \in \Rs $ as following,
\begin{equation} \label{func: shrinkage identity k} 
\left\lbrace
\begin{array}{ll} 
\vspace{1mm}
I_k^s(i,i) = 1,  & \textit{the first k diagonal elements};  \\
I_k^s(i,j) = 0,  & \textit{others}. \\
\end{array}
\right. 
\end{equation}
Operator $\trk(\cdot)$ is defined as the first $k$ partial trace of a matrix,
\begin{equation} \label{func: partial trace}
\trk(X) = \sum \limits_{i = 1}^{k}  X_{i,i}.
\end{equation}

The following matrix functions will be used in the proof of next section, and 
we want to write them out first here for reference:
\begin{equation} \label{nota: C B}
\begin{array}{l}  
\vspace{1mm}
C_{\lambda}(X) = \frac{1}{2}\|\A (X) - b \|_2^2 + \lambda T(X) ;\\  
\vspace{1mm}
C_{\lambda,\mu}(X,Z)  
	 = \mu \left\{\  C_{\lambda}(X) - \frac{1}{2}\| \A(X) - \A(Z) \|_2^2 \ \right\} 
           + \frac{1}{2}\| X - Z\|_F^2    \\ 
\vspace{1mm}
\hspace{1.7cm}	 = \mu \lambda T(X) +  \frac{\mu}{2}   \| b \|_2^2  
- \frac{\mu}{2}  \| \A(Z) \|_2^2  - \mu (\A(X), b - \A(Z)) 
           + \frac{1}{2}\| X - Z\|_F^2  ;  \\ 
B_{\mu} (Z)  = Z + \mu \A^*(b - \A(Z)).
\end{array}
\end{equation}

\section{TS1 minimization and thresholding representation}
First, let us introduce Transformed Schatten-1 penalty function(TS1) based on the singular values 
of a matrix:
\begin{equation} \label{equ: TS1}
T(X) = \sum \limits_{i = 1}^{ \rank(X) } \rho_a(\sigma_i),
\end{equation}
where $\rho_a(\cdot)$ is a linear-to-linear rational function with parameter $a \in (0,\infty)$ 
\cite{Threshold-TL1,DCATL1},
\begin{equation}  \label{equ: TL1}
\rho_a(|x|) =  \frac{(a+1)|x|}{a+|x|}.
\end{equation}

With the change of parameter $a$,  TL1 interpolates $l_0$ and $l_1$ norms:   
\[ 
	\lim_{a \to 0^{+}} \rho_a(x) = I_{\{x \neq 0\}} , \  \
 	\lim_{a \to +\infty} \rho_a(x) = |x|. 
\]
In Fig$.\;$\ref{gp: level lines}, level lines of TL1 on the plane are shown   
at small and large values of parameter $a$, resembling  
those of $l_1$ (at $a = 100$), $l_{1/2}$ (at $a=1$), and 
$l_0$ (at $a=0.01$). 

\begin{figure}
\def\arraystretch{3}
\begin{tabular}{l r}
\begin{minipage}[t]{0.45\linewidth}
\includegraphics[scale=0.35]{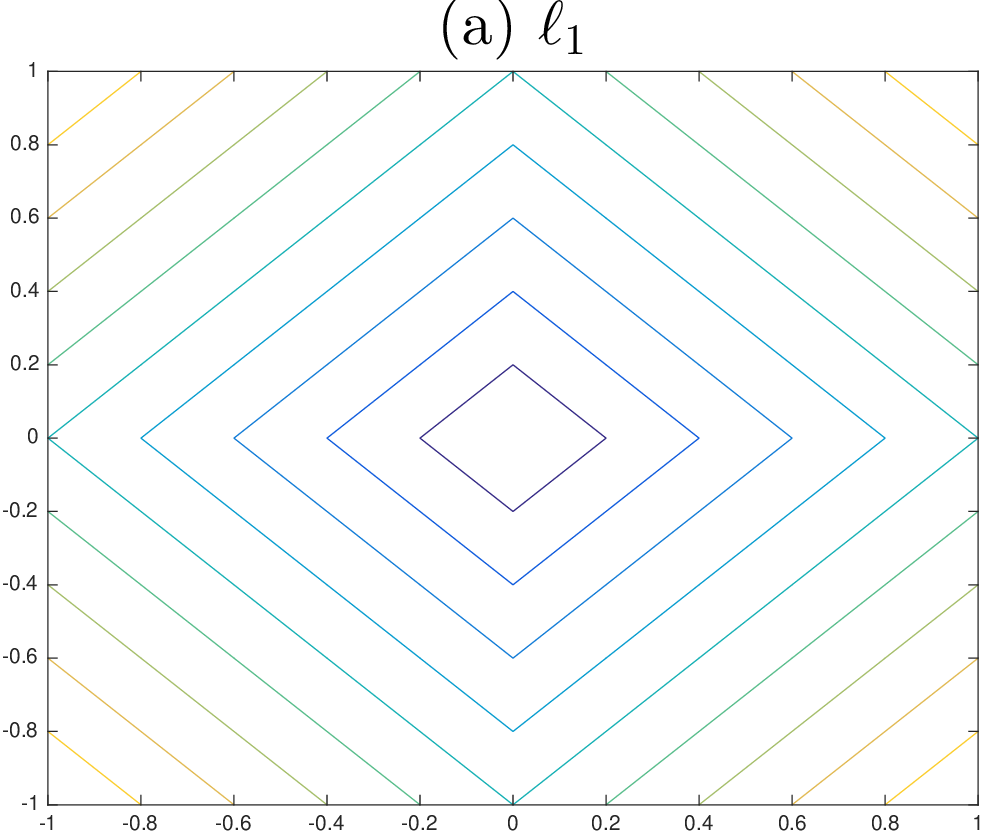} 
\end{minipage} & 
\begin{minipage}[t]{0.45\linewidth}
\includegraphics[scale= 0.35]{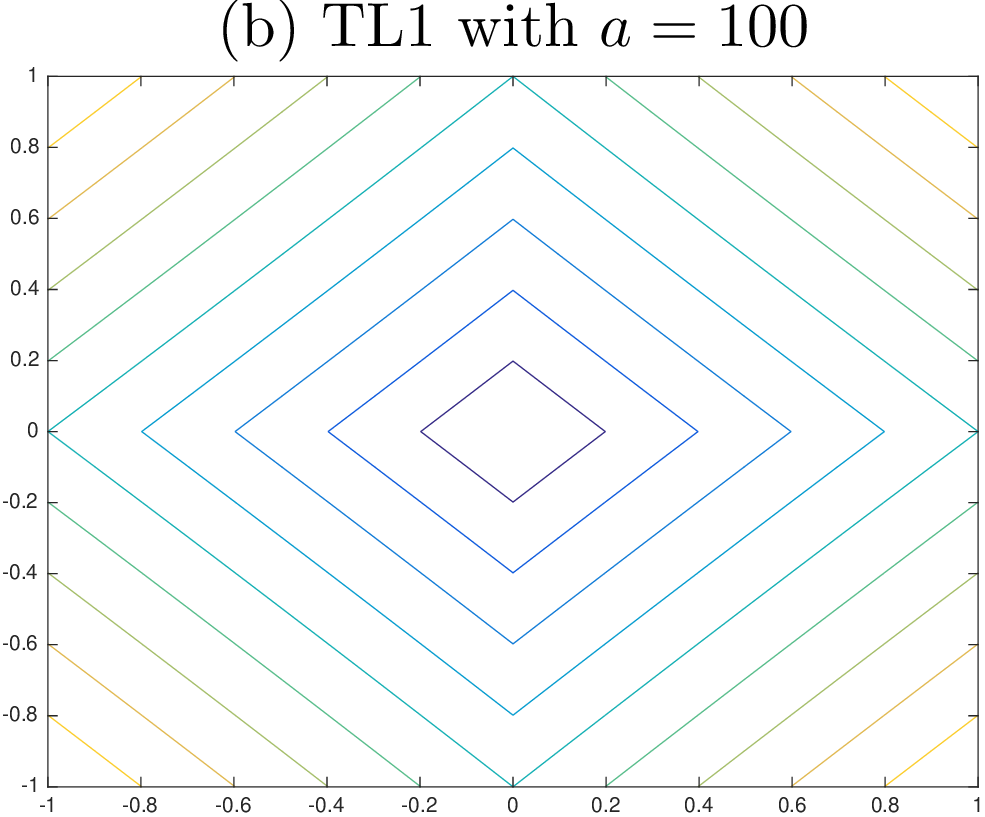} 
\end{minipage} \\
\begin{minipage}[t]{0.45\linewidth}
\includegraphics[scale=0.35]{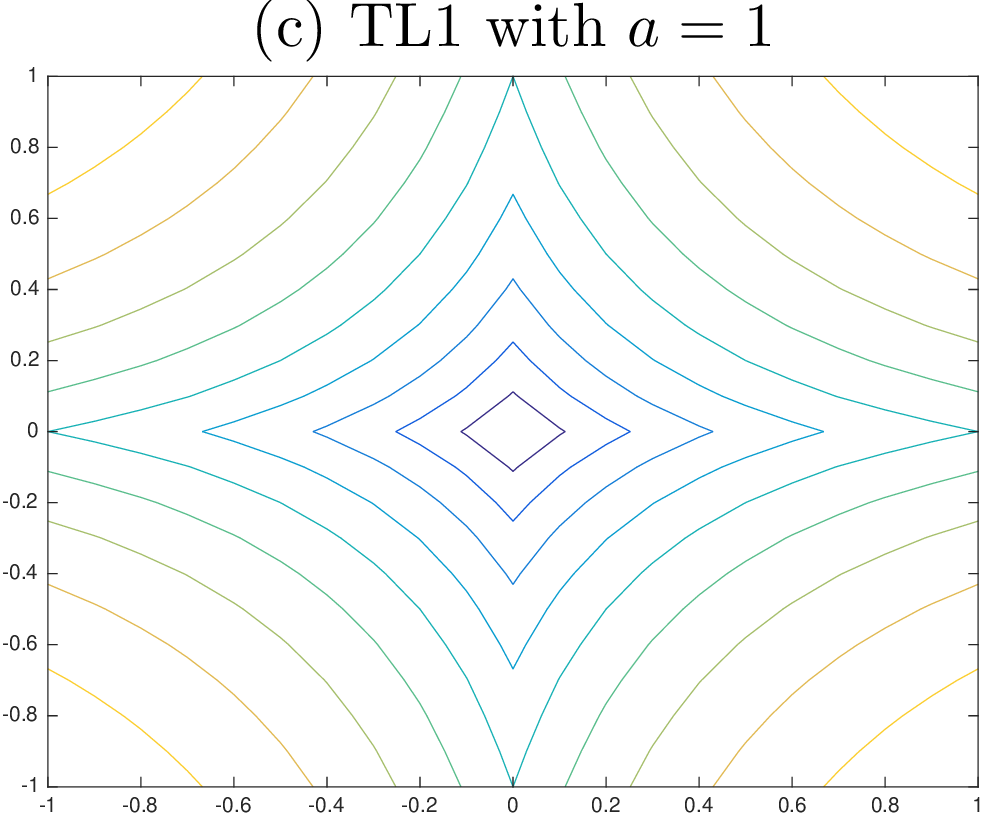} 
\end{minipage} & 
\begin{minipage}[t]{0.45\linewidth}
\includegraphics[scale= 0.35]{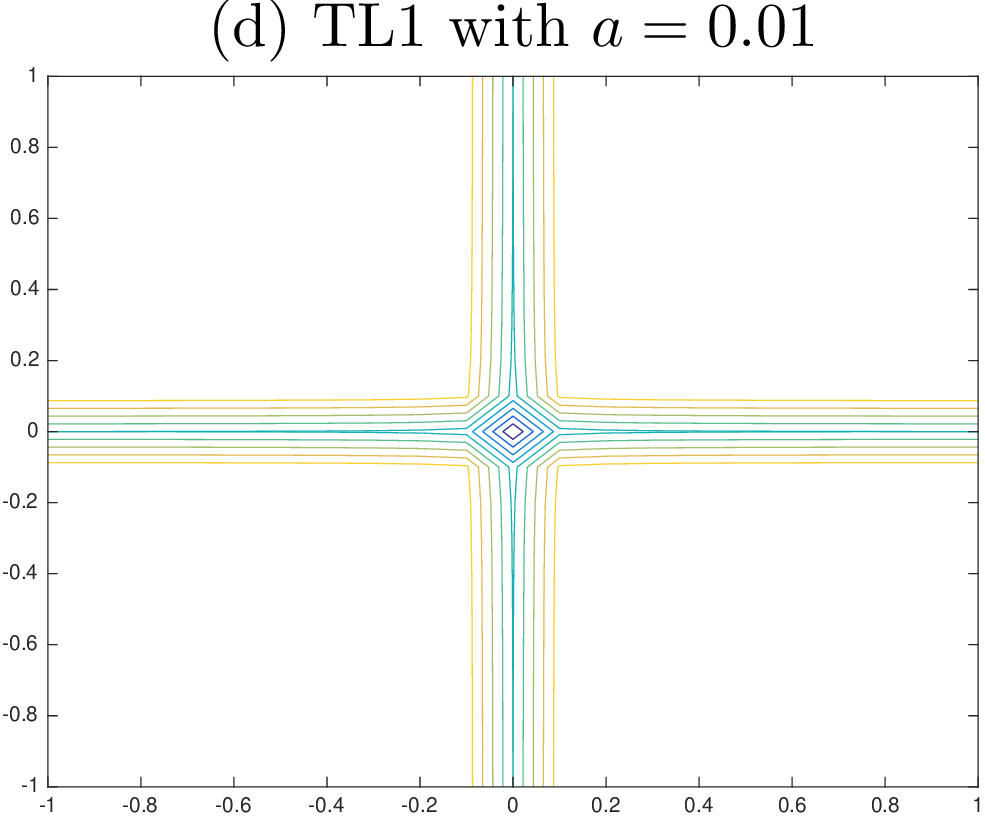} 
\end{minipage}
\end{tabular} 
\caption{Level lines of TL1 with different parameters: $a=100$ (figure b), $a=1$ (figure c), 
$a=0.01$ (figure d). For large parameter a, the graph looks almost the same as $l_1$ (figure a). 
While for small value of a, it tends to the axis.}
\label{gp: level lines}
\end{figure}

We shall focus on TS1 regularized problem
\begin{equation} \label{mini: ts1 regular}
  \min \limits_{X \in \Re^{m \times n}} { \frac{1}{2}\|\A(X) - b\|_2^2 + \lambda T(X) },
\end{equation}
where the linear transform $\A: \Re^{m \times n} \rightarrow \Re^p$ can be determined by $p$ 
given matrices $A_1,...,A_p \in \Re^{m \times n}$, that is, 
$\A(X) = \left( \langle A_1, X \rangle,...,\langle A_p, X \rangle \right)^T$.

\subsection{Overview of TL1 minimization}
To set the stage for the discussion of the TS1 regularized problem (\ref{mini: ts1 regular}), 
we review the following results on one-dimensional TL1 optimization \cite{Threshold-TL1}.

Let us consider the unconstrained TL1 regularized problem: 
\begin{equation} \label{mini: TL1 regular}
\min \limits_{x \in \Re^{n}} \frac{1}{2} \| Ax -y \|_2^2 + \lambda P_a(x),
\end{equation}
where matrix $A \in \Re^{m \times n}$, vector $ y \in \Re^{m}$ are given, 
$P_a(x) = \sum \limits_{i} \rho_a(|x_i|) $ and function $\rho_a(\cdot)$ 
is as in (\ref{equ: TL1}). 

In this subsection of TL1 minimization, we want to overwrite operator $B_{\mu}(\cdot)$ 
over vector $x$, instead of matrix field as before in (\ref{nota: C B}),
\begin{equation} \label{func: R&B}
B_{\mu}(x) = x + \mu A^T(y- Ax).
\end{equation}

In the following theorem (\ref{them: TL1 g formula}), we prove that there exists 
a closed form expression for proximal operator $prox_{\lambda \rho_a}$ on 
univariate TL1 regularization problem, where 
$
	prox_{\lambda \rho_a}(x) = \argmin \limits_{y \in \Re} 
		\frac{1}{2}(y-x)^2 + \lambda \rho_a(y).
$

Proximal operator of a convex function usually intends to solve a small convex 
regularization problem, which often admits closed-form formula or an efficient
specialized numerical methods. However, for non-convex functions, like $l_p$
with $p \in (0.1)$, their related proximal operators do not have closed form 
solutions in general. There are many iterative algorithms to approximate 
optimal solution. But they need more computing time and sometimes only 
converge to local optimal or stationary point. 
In this subsection, we prove that for TL1 function, there indeed exists a 
closed-formed formula for its optimal solution.

Different with other thresholding operators, TL1 has 2 threshold value 
formulas depending on regular parameter $\lambda$ and TL1 parameter `a'. 
We present them here with same notation as \cite{Threshold-TL1}.
\begin{equation} \label{form: threshold parameters}
\left\lbrace \begin{array}{l}
   \vspace{1mm}
   t^*_2 = \lambda \frac{a+1}{a} 
   \ \ \textit{(sub-critical parameter)} \\
   t^*_3 = \sqrt{2\lambda (a+1)} - \frac{a}{2}
   \ \ \textit{(super-critical parameter)}.
\end{array} \right.
\end{equation}
The inequality $t^*_3 \leq t^*_2$ holds and the equality is realized if and only if 
$\lambda = \frac{a^2}{2(a+1)}$, see \cite{Threshold-TL1}. 

Let $\mathrm{sgn}(\cdot)$ be the standard signum function with $\mathrm{sgn}(0)=0$, 
and 
\begin{equation}  \label{func: h formula,  part of g}
h_{\lambda}(x) = \mathrm{sgn}(x) \left\{ \frac{2}{3}(a+|x|)\, \cos\left (\frac{\varphi(x)}{3}\right ) 
                                 -\frac{2a}{3} + \frac{|x|}{3} \right\}
\end{equation}
with $ \varphi(x) = \arccos( 1 - \frac{27\lambda a(a+1)}{2(a+|x|)^3} ) $. 
In general, $| h_{\lambda}(x) | \leq |x|$, see \cite{Threshold-TL1}.

\begin{theo}(\cite{Threshold-TL1}) \label{them: TL1 g formula} 
The optimal solution of 
$
	y^* = \argmin \limits_{y \in \Re} 
		\{  \frac{1}{2}(y-x)^2 + \lambda \rho_a(|y|)  \}
$ 
is a thresholding function of the form: 
\begin{equation} 
y^* = \left\{
\begin{array}{ll}
0 ,             & |x| \leq t \\
h_{\lambda}(x), & |x| > t
\end{array} \right.
\end{equation}
where $h_{\lambda}(\cdot)$ is defined in (\ref{func: h formula,  part of g}), 
and the threshold parameter $t$ depends on $\lambda$ as follows:  
\begin{enumerate}
\item  if $\lambda \leq \frac{a^2}{2(a+1)}$ (sub-critical and critical),
$$
	t = t^*_2 = \lambda \frac{a+1}{a};
$$
\item  if $\lambda > \frac{a^2}{2(a+1)}$ (super-critical),
$$
	t = t^*_3 =  \sqrt{2\lambda (a+1)} - \frac{a}{2}. 
$$
\end{enumerate}

\end{theo}

According to the above theorem, we introduce thresholding operator 
$g_{\lambda,a}(\cdot)$ in $\Re$,
\begin{equation}  \label{func: g formula}
g_{\lambda,a}(w) = \left\{
\begin{array}{ll}
0,              & \ \text{if} \ |w| \leq t; \\
h_{\lambda}(w), & \ \text{if} \ |w| > t,
\end{array}
\right.
\end{equation}
where $t$ is the thresholding value in Theorem \ref{them: TL1 g formula} 
and $h_{\lambda}(\cdot)$ in (\ref{func: h formula,  part of g}).

In \cite{Threshold-TL1}, the authors proved that when  $\lambda < \frac{a^2}{2(a+1)}$, 
the TL1 threshold function is continuous, same as soft-thresholding function
\cite{soft-threshold-lp-daubechies2004iterative, Don_95}. 
While if $\lambda > \frac{a^2}{2(a+1)}$, the TL1 thresholding function has a 
jump discontinuity at threshold value, similar to half-thresholding function \cite{xian-half}.  
For different threshold scheme, it is believed that continuous formula is more
stable, while discontinuous formula separates nonzero and trivial coefficients
more efficiently and sometimes converges faster.

We have the following representation theorem for TL1 regularized 
problem (\ref{mini: TL1 regular}).
\medskip

\begin{theo}(\cite{Threshold-TL1})
If $x^* = (x^*_1,x_2^*,...,x^*_n)^T$ is a TL1 regularized solution 
(\ref{mini: TL1 regular}) 
with $a$ and $\lambda$ being positive constants, 
and $0 < \mu < \|A\|^{-2}$, then letting
$
	t = t^*_2 I_{\left\{ \lambda \mu \leq \frac{a^2}{2(a+1)} \right\}} 
		+ t^*_3 I_{\left\{ \lambda \mu > \frac{a^2}{2(a+1)} \right\}} 
$, 
the optimal solution satisfies the fixed point equation:
\begin{equation}
x_i^* = g_{\lambda \mu, a }([B_{\mu}(x^*)]_i) \ \ \ \forall \ i = 1,...,n.
\end{equation}
\end{theo}

In the following, we will extend this result to TS1 low rank matrix completion
and propose 2 thresholding algorithms based on it.

\subsection{TS1 thresholding representation theory}
Here we assume $m \leq n$. 
For a matrix $X \in \Rs$ with rank equal to $r$, its singular values vector
$\sigma = (\sigma_1,...,\sigma_m)$
is arranged as
$$
\sigma_1 \geq \sigma_2 \geq ... \geq \sigma_r > 0 = \sigma_{r+1} = ... = \sigma_m.
$$
The singular value decomposition (SVD) is $X = U D V^T$, 
where $U = (U_{i,j})_{m\times m}$ and $V = (V_{i,j})_{n\times n}$ are unitary matrices, 
with $D = Diag(\sigma) \in \Rs$ diagonal.

In \cite{KFan_inequality}, Ky Fan proved the dominance theorem and derive the following 
Ky Fan k-norm inequality. 
\medskip

\begin{lem} (Ky Fan $k$-norm inequality) \label{lem: Ky Fan inequ}
For a matrix $X \in \Rs$ with SVD: $X = U\, D\, V^T$, where diagonal elements of $D$ are
arranged in decreasing order, we have:
\begin{equation*}
\langle X, I_k^s \rangle \leq \langle D, I_k^s \rangle,
\end{equation*}
that is, $\trk(X) \leq \trk(D) = \|X\|_{Fk} $, $\forall k = 1,2,...,m$.
The inequalities become equalities if and only if $X=D$. 
Here matrix $I_k^s$ and operator $\trk(\cdot)$ are defined in section \ref{section: notation}.
\end{lem}

\medskip
Another proof of this inequality without using dominance theorem is available.
We leave it in the appendix for readers' convenience, making the paper self-contained. 
\medskip

\begin{theo} \label{them: TS1 repres}
For any matrix $Y \in \Rs$, which admits a singular value decomposition: 
$Y = U\, \Diag(\sigma) \, V^T$, where $ \sigma = (\sigma_1,...,\sigma_m)$. 
A global minimizer of  
$\min \limits_{X \in \Re^{m \times n}} { \frac{1}{2}\| X - Y \|_F^2 + \lambda T(X) }$ 
is: 
\begin{equation}
X^s = G_{\lambda,a}(Y) = U \Diag(g_{\lambda,a}(\sigma))V^T, 
\end{equation}
where $g_{\lambda,a}(\cdot)$ is defined in (\ref{func: g formula}) and applied entrywise to $\sigma$.
\end{theo}
\medskip

\begin{proof}
First due to the unitary invariance property of Frobenius norm and $Y = U \Diag(\sigma)V^T$, we have 
\begin{equation*}
\frac{1}{2}\| X - Y \|_F^2 + \lambda T(X) = \frac{1}{2}\| U^T X V - \Diag(\sigma) \|_F^2 + \lambda T(U^T X V).
\end{equation*}
So
\begin{equation} \label{equ: Xs bridge to D}
\begin{array}{ll}
\vspace{2mm}
X^s & = \argmin \limits_{X \in \Rs} \frac{1}{2}\| X - Y \|_F^2 + \lambda T(X) \\
    & = U  \left\{ \argmin \limits_{X \in \Rs} \frac{1}{2}\| X - Diag(\sigma) \|_F^2 
        + \lambda T(X) \right\} V^T. 
\end{array}    
\end{equation}
Next we want to show: 
\begin{equation} \label{equ: X to D}
\begin{array}{l}
\vspace{2mm}
\argmin \limits_{X \in \Rs} \frac{1}{2}\| X - \Diag(\sigma) \|_F^2 + \lambda T(X) \\
= \argmin_{ \{ D \in \Rs \ \text{is diagonal} \} } \frac{1}{2}\| D - \Diag(\sigma) \|_F^2 + \lambda T(D)
\end{array}    
\end{equation}

For any $X \in \Rs$, suppose it admits SVD: $X =  U_x Diag(\sigma_{x}) V_x^T$. 
Denote 
$$
D_x = \Diag(\sigma_{x}) \ \text{and} \ D_y = \Diag(\sigma).
$$
We can rewrite diagonal matrix $D_y$ as $D_y = \sum \limits_{i}^{m} \triangledown \sigma_i I_i^s$, 
where $\triangledown \sigma_i = \sigma_i - \sigma_{i+1} \geq 0$ for $i = 1,2,...,m-1$, and 
$\triangledown \sigma_m = \sigma_m$.  
So simply, $\sum \limits_{i = k}^{m} \triangledown \sigma_i = \sigma_k$. 
Note the shrinkaged identity matrix $I^s_i$ is defined in section \ref{section: notation}. 
\begin{equation*}
\begin{array}{ll}
\langle X, D_y \rangle 
& = \langle X, \sum \limits_{i}^{m} \triangledown \sigma_i I_i^s \rangle 
   = \sum \limits_{i}^{m} \langle X, \triangledown \sigma_i I_i^s \rangle \\
& \leq \sum \limits_{i}^{m} \langle D_x , \triangledown \sigma_i I_i^s \rangle 
   = \langle D_x, D_y \rangle,                     
\end{array}
\end{equation*}
where we used Lemma \ref{lem: Ky Fan inequ} for the inequality. 
The equality holds if and only if $X = D_x$. 

Thus we have 
\begin{equation*}
\begin{array}{ll}
\| X - D_y \|_F^2 
& = \| X \|_F^2 + \| D_y \|_F^2 - 2 \langle X, D_y \rangle \\
& \geq \| D_x \|_F^2 + \| D_y \|_F^2 - 2 \langle D_x , D_y \rangle = \|D_x - D_y\|_F^2.
\end{array}    
\end{equation*}
Also due to $T(X) = T(D_x)$,
\[
\frac{1}{2}\| X - D_y \|_F^2 + \lambda T(X) 
\geq \frac{1}{2}\| D_x - D_y \|_F^2 + \lambda T(D_x). 
\]
Only when $X = D_x$ is a diagonal matrix, the above will become equality. 
So we finish the proof of equation (\ref{equ: X to D}).

Denote a diagonal matrix $D  \in \Rs$ as $D = \Diag(d)$. Then:
\[ \frac{1}{2}\| D - \Diag(\sigma) \|_F^2 + \lambda T(D) 
= \sum \limits_{i=1}^{m} \left\{ \frac{1}{2}\| d_i - \sigma_i \|_2^2 + \lambda \rho_a(|d_i|) \right\}
\]
By Theorem \ref{them: TL1 g formula}, we have 
$ g_{\lambda,a}(\sigma_i) = arg \min \limits_{d} \{ \ \frac{1}{2}\| d - \sigma_i \|_2^2 
+ \lambda \rho_a(|d|) \ \} \geq 0$. 
It follows that
\begin{equation} 
\begin{array}{l}
\vspace{1mm}
\argmin_{ \{ D \in \Rs \text{and D is diagonal} \} } \frac{1}{2}\| D - \Diag(\sigma) \|_F^2 
  + \lambda \ T(D) \\ \vspace{1mm}
= \argmin \limits_{X \in \Rs} \frac{1}{2}\| X - \Diag(\sigma) \|_F^2 + \lambda \ T(X) \\ 
= \Diag(g_{\lambda,a}(\sigma)).
\end{array}    
\end{equation}

In view of (\ref{equ: Xs bridge to D}), the matrix 
$X^s = U Diag(g_{\lambda,a}(\sigma))V^T$ is a global minimizer, 
which will be denoted as $G_{\lambda,a}(Y)$.
The proof is complete.
\end{proof}
\medskip

\begin{lem}  \label{lemma: C-(X,Z) }
For any fixed $\lambda > 0$, $\mu >0$ and matrix $Z \in \Re^{m \times n}$,
let $X^s = G_{\lambda \mu,a}(B_{\mu}(Z))$,
then for any matrix $X \in \Re^{m \times n}$, 
$$
C_{\lambda,\mu}(X^s,Z) \leq C_{\lambda,\mu}(X,Z),
$$
which means $X^s$ is a global minimizer of $C_{\lambda,\mu}(X,Z)$.
Here the matrix function $C_{\lambda,\mu}(X,Z)$ is defined in (\ref{nota: C B}) 
of section \ref{section: notation}.
\end{lem}

\begin{proof}
First, we will rewrite the formula of $C_{\lambda,\mu}(X,Z)$. 
Note that $\A(X)$ and $\A(Z)$ are vectors in space $\Re^p$. 
Thus in the formula of $C_{\lambda,\mu}(X,Z)$, there exist norms and inner products  
for both matrices and vectors. 
By definition, 
\begin{equation}
\begin{array}{lll}
C_{\lambda,\mu}(X,Z) 
& = & \frac{1}{2} \|X\|_F^2 - \langle X,Z \rangle + \frac{1}{2}\|Z\|_F^2 
      + \lambda \, \mu\,  T(X) + \frac{\mu}{2} \, \|b\|^2_2 \\  \vspace{1mm}
& &   - \mu \, (\A(X), b - \A(Z) ) -\frac{\mu}{2}\, \|\A(Z)\|^2_2 \\                

& = & \frac{1}{2} \|X\|_F^2 + \frac{1}{2}\|Z\|_F^2  
      + \frac{\mu}{2} \|b\|^2_2 -\frac{\mu}{2}\|\A(Z)\|^2_2 \\  \vspace{1mm}
& &   + \lambda \, \mu \, T(X) - \langle \ X, Z + \mu \A^* (b - \A(Z)) \ \rangle \\

& = & \frac{1}{2}\| X - B_{\mu}(Z) \|_F^2 + \lambda \, \mu \,  T(X) \\
& &   -\frac{1}{2}\|B_{\mu}(Z)\|_F^2 + \frac{1}{2}\|Z\|_F^2  
      + \frac{\mu}{2} \|b\|^2_2 -\frac{\mu}{2}\|\A(Z)\|^2_2 
\end{array}
\end{equation}
Thus if we fix matrix $Z$,
\begin{equation}
\begin{array}{l}
\argmin \limits_{X \in \Re^{m \times n}} C_{\lambda,\mu}(X,Z)  
= \argmin \limits_{X \in \Re^{m \times n}} \frac{1}{2}\| X - B_{\mu}(Z) \|_F^2 + \lambda \mu T(X)
\end{array}
\end{equation}
Then by Theorem \ref{them: TS1 repres}, $X^s$ is a global minimizer. 
\end{proof}
\medskip

\begin{theo} \label{lemma: C(X) - C(X<Z)} 
For fixed parameters, $\lambda>0$ and $0 < \mu < \|\A\|^{-2}_2$. 
If $X^*$ is a global minimizer for problem $C_{\lambda}(X)$, 
then $X^*$ is also a global minimizer for problem $\min \limits_{X \in \Rs} C_{\lambda,\mu}(X,X^*)$, 
that is 
\begin{equation*}
C_{\lambda,\mu}(X^*,X^*) \leq C_{\lambda,\mu}(X,X^*), \ \ \ \forall X \in \Re^{m \times n}.
\end{equation*}
\end{theo}

\begin{proof}
\begin{equation}
\begin{array}{rl}
C_{\lambda,\mu}(X,X^*) 
= & \mu \{ \frac{1}{2}\|\A(X) - b\|^2_2 + \lambda T(X) \}  \\
&   + \frac{1}{2}\{ \|X -X^*\|_F^2 - \mu\|\A(X) - \A(X^*)\|^2_2 \} \\
\geq & \mu \{ \ \frac{1}{2} \|\A(X) - b\|^2_2 + \lambda T(X) \ \} 
= \mu C_\lambda(X) \\
\geq & \mu C_\lambda(X^*) =  C_{\lambda,\mu}(X^*,X^*)
\end{array}
\end{equation}
The first inequality is due to the fact:
\begin{equation}
\begin{array}{ll}
\| \A(X) - \A(X^*) \|_2^2 
& = \| A \vc(X) - A \vc(X^*) \|_2^2 \\
& \leq \|A\|_2^2  \ \|\vc(X-X^*)\|_2^2 \\
& \leq \|\A\|_2^2 \ \|X-X^*\|_F^2
\end{array}
\end{equation}
\end{proof}

By the above Theorems and Lemmas, if $X^*$ is a global minimizer of $C_\lambda(X)$, it is also a 
global minimizer of $C_{\lambda, \mu}(X,Z)$ with $Z=X^*$, which has a closed form solution formula. 
Thus we arrive at the following fixed point equation for the 
global minimizer $X^*$: 
\begin{equation} \label{equ: fix point, global mini}
X^* = G_{\lambda \mu,a}(B_\mu(X^*)).
\end{equation}
Suppose the SVD for matrix $B_\mu(X^*)$ is $ U \, \Diag(\sigma_b^*) \, V^T $,
then 
$$ 
X^* = U \, \Diag( g_{\lambda \mu,a}(\sigma_b^*)) \, V^T,
$$
which means that the singular values of $X^*$ satisfy $\sigma^*_i = g_{\lambda \mu,a}(\sigma_{b,i}^*)$, for $i = 1,...,m$.

\section{TS1 thresholding algorithms}
Next we will utilize fixed point equation (\ref{equ: fix point, global mini}) to  derive two thresholding 
algorithms for TS1 regularized problem (\ref{mini: ts1 regular}). 
As in \cite{Threshold-TL1,DCATL1}, from the equation 
$X^* = G_{\lambda \mu,a}(B_\mu(X^*)) = U \Diag(g_{\lambda \mu,a}(\sigma)) V^T$,
we will replace optimal matrix $X^*$ with $X^k$ on the left and $X^{k-1}$ on the 
right at the $k$-th step of iteration as: 
 \begin{equation}
\begin{array}{ll}
X^k & = G_{\lambda \mu,a}(B_\mu(X^{k-1})) \\
    & = U^{k-1} \, \Diag\left( g_{\lambda \mu,a}( \sigma^{k-1}) \right) \, V^{k-1,T},
\end{array}
\end{equation}
where unitary matrices $U^{k-1}$, $V^{k-1}$ and singular values $\{ \sigma^{k-1} \}$ 
come from the SVD decomposition of matrix $B_{\mu}(X^{k-1})$.  
Operator $g_{\lambda \mu,a}(\cdot)$ is defined in (\ref{func: g formula}), and 
\begin{equation}
g_{\lambda \mu,a}(w) = \left\{
\begin{array}{ll}
0, & {\rm if} \; |w| < t; \\
h_{\lambda \mu}(w), & {\rm if} \;  |w| \geq t.
\end{array}
\right.
\end{equation}

Recall that the thresholding parameter $t$ is: 
\begin{equation}  \label{form: optiaml t}
t = \left\{
\begin{array}{ll}
       t^*_2 = \lambda \mu \frac{a+1}{a},                & \text{if} \ \lambda \leq \frac{a^2}{2(a+1)\mu}; \\
       t^*_3 = \sqrt{2 \lambda \mu (a+1)} - \frac{a}{2}, & \text{if} \ \lambda > \frac{a^2}{2(a+1)\mu}.
\end{array}
\right.
\end{equation}

With an initial matrix $X^0$, we obtain an iterative algorithm, 
called TS1 iterative thresholding (IT) algorithm. It is the basic TS1 iterative scheme. 
Later, two adaptive and more efficient IT algorithms (TS1-s1 and TS1-s2) will be introduced. 

\subsection{Semi-Adaptive Thresholding Algorithm -- TS1-s1}
We begin with formulating an optimal condition for regularization parameter $\lambda$, 
which serves as the basis for the parameter selection and updating in this 
semi-adaptive algorithm. 

Suppose optimal solution matrix $X$ has rank $r$, by prior knowledge or estimation. 
Here, we still assume  $m \leq n$. 
For any $\mu$, denote $B_{\mu}(X) = X + \mu A^T( b-\A(X) )$ and  $\{ \sigma_i \}_{i=1}^{m}$ 
are the $m$ non-negative singular values for $B_{\mu}(X)$.

Suppose that $X^*$ is the optimal solution matrix of (\ref{mini: ts1 regular}), 
and the singular values of matrix $B_{\mu}(X^*)$ are denoted as 
$ \sigma^*_1 \geq \sigma^*_2 \geq ... \geq \sigma^*_m $. 
Then by the fixed equation (\ref{equ: fix point, global mini}), 
the following inequalities hold: 
\begin{equation}
\begin{array}{lll}
\sigma^*_i > t     & \Leftrightarrow &  i \in \{ 1,2,...,r \}, \\
\sigma^*_j \leq t & \Leftrightarrow &  j \in \{ r+1,r+2,...,m \}, 
\end{array}
\end{equation}
where $t$ is our threshold value. Recall that $t^*_3 \leq t \leq t^*_2$. So 
\begin{equation}
\begin{array}{l}
\sigma^*_r \geq t \geq t^*_3 = \sqrt{2\lambda \mu (a+1)} - \frac{a}{2}; \\
\sigma^*_{r+1} \leq  t \leq t^*_2 = \lambda \mu \frac{a+1}{a}. 
\end{array}
\end{equation}
It follows that
\begin{equation*}
\lambda_1 \equiv \dfrac{a \sigma^*_{r+1}}{\mu (a+1)} \leq \lambda 
\leq \lambda_2 \equiv \dfrac{(a+2 \sigma^*_{r})^2}{8(a+1)\mu} 
\end{equation*}
or $\lambda^* \in [\lambda_1, \lambda_2]$.

The above estimate helps to set optimal regularization parameter. 
A choice of $\lambda^*$ is
\begin{equation} \label{equ: lambda}
\lambda^* = \left\{
\begin{array}{ll}
\lambda_1, & \quad \text{if} \ \ 
   \lambda_1 \leq \frac{a^2}{2(a+1)\mu}, \ \ \text{then} \ \ \lambda^* \leq \frac{a^2}{2(a+1)\mu} 
                                         \Rightarrow t = t^*_2; \\
\lambda_2, & \quad \text{if} \ \ 
   \lambda_1 > \frac{a^2}{2(a+1)\mu},    \ \ \text{then} \ \ \lambda^* > \frac{a^2}{2(a+1)\mu}
                                        \Rightarrow t = t^*_3.
\end{array} \right.
\end{equation}

In practice, we approximate $B_{\mu}(X^*)$ by $B_{\mu}(X^n)$ in  (\ref{equ: lambda}), so 
$ \lambda_1 = \dfrac{a \sigma^*_{r+1}}{\mu (a+1)}$, and
$ \lambda_2 = \dfrac{(a+2 \sigma^*_{r})^2}{8(a+1)\mu}$.
We choose optimal parameter $\lambda$ at the $n$-th step as 
\begin{equation} \label{form: lamba in s2}
\lambda^*_n = \left\{
\begin{array}{ll}
\lambda_1, & \quad \text{if} \ \ \lambda_1 \leq \frac{a^2}{2(a+1)\mu},\\
\lambda_2, & \quad \text{if} \ \ \lambda_1 > \frac{a^2}{2(a+1)\mu}.
\end{array} \right.
\end{equation}

This way, we obtain an adaptive iterative algorithm without pre-setting the  
regularization parameter $\lambda$. The TL1 parameter $a$ is still free 
and needs to be selected beforehand. 
Thus the algorithm is overall semi-adaptive, called TS1-s1 for short and summarized 
in Algorithm \ref{algorithm: TS1-s1}.

\begin{algorithm}[H]
\caption{TS1-s1 threshold algorithm} 
\label{algorithm: TS1-s1} 
\begin{algorithmic}  
\STATE Initialize: \ \ Given $X^0$ and parameter $\mu$ and $a$.  
\WHILE{NOT converged} 
  \STATE  1. $Y^n = B_{\mu}(X^n) = X^n - \mu \A^*(\A(X^n)-b)$, \\
         \ \ \ and compute SVD of $Y^n$ as $Y^n = U \, \Diag(\sigma) \, V^T$ ;  \\ 
  \STATE  2. Determine the value for $\lambda^n$ by (\ref{form: lamba in s2}), \\ 
         \ \ \ then obtain related threshold value $t^n$ by (\ref{form: optiaml t}); \\       
  \STATE  3. $X^{n+1} = G_{\lambda^n \mu, a}(Y^n) =  U \Diag(g_{\lambda^n \mu, a}(\sigma)) V^T$; \\
  \STATE  Then, $n \rightarrow n+1.$
\ENDWHILE 
\end{algorithmic} 
\end{algorithm} 

\subsection{Adaptive Thresholding Algorithm -- TS1-s2}
Different from TS1-s1 where the parameter '$a$' needs to be determined manually, here at each 
iterative step, we choose $a=a_n$ such that equality  
$\lambda_n = \frac{a^2_n}{2(a_n +1)\mu_n}$ holds. 
The threshold value $t$ is given by a single formula with $t = t^*_3 = t^*_2$.
 
Putting $\lambda = \frac{a^2}{2(a +1) \mu}$ at critical value, 
the parameter $a$ is expressed as: 
\begin{equation} \label{form: scheme 2 a lambda}
a = \lambda\mu + \sqrt{(\lambda \mu)^2 + 2 \lambda \mu }.
\end{equation}
The threshold value is: 
\begin{equation} \label{form: scheme 2 t}
t  = \lambda \mu \frac{a+1}{a} 
   = \frac{\lambda\mu}{2} + \frac{\sqrt{(\lambda \mu)^2 + 2 \lambda \mu }}{2}.
\end{equation}

Let $X^*$ be the TL1 optimal solution and $\sigma^*$ be the singular values for matrix $B_{\mu}(X^*)$. 
Then we have the following inequalities: 
\begin{equation}
\begin{array}{lll}
\sigma^*_i > t & \Leftrightarrow &  i \in \{ 1,2,...,r \}, \\
\sigma^*_j \leq t & \Leftrightarrow &  j \in \{ r+1,r+2,...,m \}.
\end{array}
\end{equation}
So, for parameter $\lambda$, we have: 
$$ \dfrac{2 (\sigma^*_{r+1})^2}{1+2 \sigma^*_{r+1}} \leq \lambda 
\leq \dfrac{2 (\sigma^*_{r})^2}{1+2 \sigma^*_{r}}.$$
Once the value of $\lambda$ is determined, the parameter $a$ is 
given by (\ref{form: scheme 2 a lambda}).

In the iterative method, we approximate the optimal solution $X^*$ by 
$X^n$ and further use $B_{\mu}(X^n)$'s singular values $\{\sigma^n_i\}_i$ 
to replace those of $B_{\mu}(X^*)$.
The resulting parameter selection is:
\begin{equation}  \label{form: lamba and a for s3}
\begin{array}{l}
\lambda_n = \dfrac{2 (\sigma^n_{r+1})^2}{1+2 \sigma^n_{r+1}}; \\
a_n = \lambda_n \mu_n + \sqrt{(\lambda_n \mu_n)^2 + 2 \lambda_n \mu_n }.
\end{array}
\end{equation}

In this algorithm (TS1-s2 for short), only parameter $\mu$ is fixed, satisfying inequality
$\mu \in (0, \|A\|^{-2})$. 
Its algorithm is summarized in Algorithm \ref{algorithm: TS1-s2}.  

\begin{algorithm}[H]
\caption{TS1-s2 threshold algorithm} 
\label{algorithm: TS1-s2} 
\begin{algorithmic}  
\STATE Initialize: \ \ Given $X^0$ and parameter $\mu$.  
\WHILE{NOT converged} 
  \STATE  1. $Y^n = X^n - \mu \A^*(\A(X^n)-b)$, and compute SVD of $Y^n$ 
                as $Y^n = U \, \Diag(\sigma) \, V^T$ ;  \\ 
  \STATE  2. Determine the values for $\lambda^n$ and $a^n$ by (\ref{form: lamba and a for s3}), \\
         \ \ \ \ then update threshold value $t^n = \lambda^n \mu \frac{a^n+1}{a^n}$;  \\       
  \STATE  3. $X^{n+1} = G_{\lambda^n \mu, a^n}(Y^n) =  U \, \Diag(g_{\lambda^n \mu, a}(\sigma)) \, V^T$; \\
  \STATE  Then $n \rightarrow n+1.$
\ENDWHILE 
\end{algorithmic} 
\end{algorithm} 

\section{Numerical experiments}
In this section, we present numerical experiments to illustrate the effectiveness of
our Algorithms: semi-adaptive TS1-s1 and adaptive TS1-s2, 
compared with several state-of-art solvers on matrix completion problems
\footnote{TS1 matlab codes can be downloaded from https://github.com/zsivine/TS1-algorithms}.
The comparison solvers include:
\begin{itemize}
\item  LMaFit  \cite{LMaFit},
\item  FPCA \cite{FPCA}, 
\item  sIRLs-q \cite{Fazel}, 
\item  IRucLq-M \cite{IRucLq},
\item  LRGeomCG \cite{vandereycken2013low}
\end{itemize}

The code LMAFit solves a low-rank factorization model, instead of computing SVD 
which usually takes a big chunk of computation time. Also part of its codes is written
in C, same as LRGeomCG. So once this method converges, it is the fastest method 
among all comparisons. All others codes are implemented under Matlab environment 
and involve SVD approximated by fast Monte Carlo algorithms 
\cite{drineas2006fastI,drineas2006fastII}. 
FPCA is a nuclear norm minimization code, while sIRLs-q and IRucLq-M are iterative 
reweighted least square algorithms for Schatten-q quasi-norm optimizations.
LRGeomCG algorithm explores matrix completion based on Riemannian optimization. 
It tries to minimize the least-square distance on the sampling set over the Riemannian 
manifold of fixed-rank matrices. When the rank information is known priori or well 
approximated, this method is efficient and accurate, as shown in these experiments 
below, especially for standard Gaussian matrices. But a drawback of LRGeomCG is that 
the rank of the manifold is fixed. Basically, it is hard for it to handle unknown rank cases. 

In our TS1 algorithms, MC SVD algorithm \cite{drineas2006fastII} is implemented at each 
iteration step, same as FPCA. We also tried another fast SVD approximation algorithms, but 
MC SVD is the most suitable one, satisfying both speed and accuracy requirements in one 
iterative algorithm. 
All our tests were performed on a $Lenovo$ desktop: 16 GB of RAM and Intel@ Core Quad 
processor $i7$-$4770$ with CPU at $3.40$GHz under 64-bit Ubuntu system.

We tested and compared these solvers on low rank matrix completion problems
under various conditions, including multivariate Gaussian, uniform and $\chi^2$ distributions.
We also tested the algorithms on grayscale image recovery from partial observations
(image inpainting).  

\subsection{Implementation details}

In the following series of tests, we generated random matrices 
$$
	M = M_L M_R^T \in \R_{m \times n},
$$ 
where matrices $M_L$ and $M_R$ are in spaces $\R^{m \times r}$ and 
$\R^{n \times r}$ respectively.

By setting parameter $r$ to be small, we obtain a low rank matrix $M$ with rank 
at most $r$. After this step, we uniformly random-sampled a subset $\omega$ 
with $p$ entries from $M$. 
The following quantities help to quantify the difficulty of a recovery problem.
\begin{itemize}
\item SR (Sampling ratio): SR = $p/mn$.
\item FR (Freedom ratio): FR = $r(m+n-r)/p$, which is the freedom of rank $r$ 
matrix divided by the number of measurement. According to \cite{FPCA} , 
if FR $>1$, there are infinite number of matrices with rank $r$ and the given entries. 
\item $r_m$ (Maximum rank with which the matrix can be recovered):
$$
	r_m = \lfloor \frac{m+n-\sqrt{(m+n)^2-4p}}{2} \rfloor  
	\ \ \textit{(floor function)},
$$ 
which is defined as the largest rank such that FR $\leq 1$.
\end{itemize}

The TS1 thresholding algorithms do not guarantee a global minimum in general, similar 
to non-convex schemes in 1-dimensional compressed sensing problems. Indeed we observe 
that TS1 thresholding with random starts may get stuck at local minima especially when 
parameter FR (freedom ratio) is high or the matrix completion is difficult. A good initial matrix
 $X^0$ is important for thresholding algorithms. In our numerical experiments, instead of 
choosing $X^0 = 0$ or random, we set $X^0$  equal to matrix $M$ whose elements are 
as observed on $\Omega$ and zero elsewhere.

The stopping criterion is
\begin{equation*}
	\dfrac{\|X^{n+1} - X^n \|_F}{ \max\{ \|X^n \|_F, 1 \} } \leq tol
\end{equation*}
where $X^{n+1}$ and $X^{n}$ are numerical results from two contiguous iterative steps, and 
$tol$ is a moderately small number.  In all these following experiments, we fix $tol = 10^{-6}$
with maximum iteration steps $1000$.

We also use the relative error
\begin{equation} \label{relative error}
	\mathrm{rel.err} = \dfrac{\|X_{opt} - M \|_F}{ \| M \|_F } 
\end{equation}
to estimate the closeness of $X_{opt}$ to $M$, where $X_{opt}$ is the "optimal" solution 
produced by all numerical algorithms.    

\subsubsection{Rank estimation}

For thresholding algorithms, rank $r$ is the most important parameter, especially for our
TS1 methods, where thresholding value $t$ is determined based on $r$. If the true rank 
$r$ is unknown, we adopt the rank decreasing estimation method 
(also called maximum eigengap method) as in \cite{IRucLq,LMaFit}, thereby
extending both TS1-s1 and TS1-s2 schemes to work with an overestimated initial rank 
parameter $K$. In the following tests, unless otherwise specified, we set 
$K = \lfloor1.5\, r\rfloor$.
The idea behind this estimation method is as follows. Suppose that at step $n$, 
our current matrix is $X$. The eigenvalues of $X^T\, X$ are arranged with descending 
order and 
$\lambda_{r_{min}} \geq \lambda_{r_{min}+1} \geq ... \geq \lambda_{K +1} > 0 $ 
is the $r_{min}$-th through $K+1$-th eigenvalues of $X^TX$, where  $r_{min}$ is 
manually specified minimum rank estimate. 
Then we compute the quotient sequence 
$ \widehat{\lambda_i} = \lambda_i/\lambda_{i+1} $, 
$i = r_{min} , ... , K$.  
Let 
$$ 
	\widetilde{K}  = \argmin \limits_{r_{min} \leq i \leq K} \widehat{\lambda_i},
$$
the corresponding index for maximal element of $\{ \widehat{\lambda_i} \}$. 
If the eigenvalue gap indicator
$$
	\tau =  \widehat{\lambda}_{\widetilde{K}} (K - r_{min} + 1)    /  \sum \limits_{i \neq \widetilde{K}} \widehat{\lambda_i}   
 \ \ 	> 10,
$$
we adjust our rank estimator from $K$ to $\widetilde{K}$. During numerical simulations,
we did this adjustment only once for each problem. In most cases, this estimation 
adjustment is quite satisfactory and the adjusted estimate is very close to the true rank $r$. 

\subsubsection{Choice of a: optimal parameter testing for TS1-s1}

A major difference between TS1-s1 and TS1-s2 is the choice of 
parameter $a$, which influences the behaviour of penalty function $\rho_a(\cdot)$ of  
TS1. When '$a$' tends to zero, the function $T(X)$ approaches the rank. 

We tested TS1-s1 on small size low rank matrix completion with different `$a$' values, 
varying among  $\{0.1, 0.5, 1, 10, 100 \}$, for both known rank scheme and 
the scheme with rank estimation. In these tests, $M = M_L M_R^T$ is a $100 \times 100$ 
random matrix, where $M_L$ and $M_R$ are generated under i.i.d  standard normal 
distribution. The rank $r$ of $M$ varies from $10$ to $22$. 

For each value of `$a$', we conducted $50$ independent tests with different $M$ and
sample index set $\omega$. We declared $M$ to be recovered successfully if the 
relative error (\ref{relative error}) was less than $5 \times 10^{-3}$. 
The test results for known rank scheme and rank estimation scheme 
are both shown in Figure \ref{figure: optimal a test TS1-s1}.   
The success rate curves of rank estimation scheme are not as clustered as those of 
known rank scheme. 
In order to clearly identify the optimal parameter '$a$', we ignored the curve 
of $a= 0.1$ in the right figure as it is always below all others. 
The vertical red dotted line there indicates the position where FR $= 0.6$. 

\begin{figure}
\begin{tabular}{lr}
\begin{minipage}[t]{0.5\linewidth}
\includegraphics[scale=0.45]{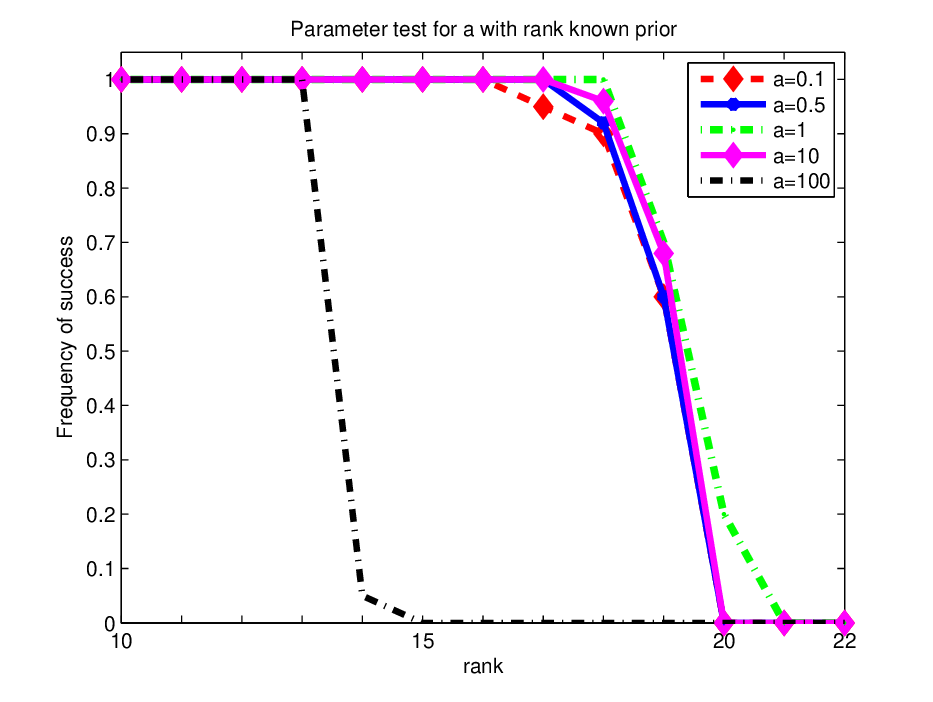}
\caption*{Rank is known a prior}   
\end{minipage}  &
\begin{minipage}[t]{0.4\linewidth}
\includegraphics[scale=0.45]{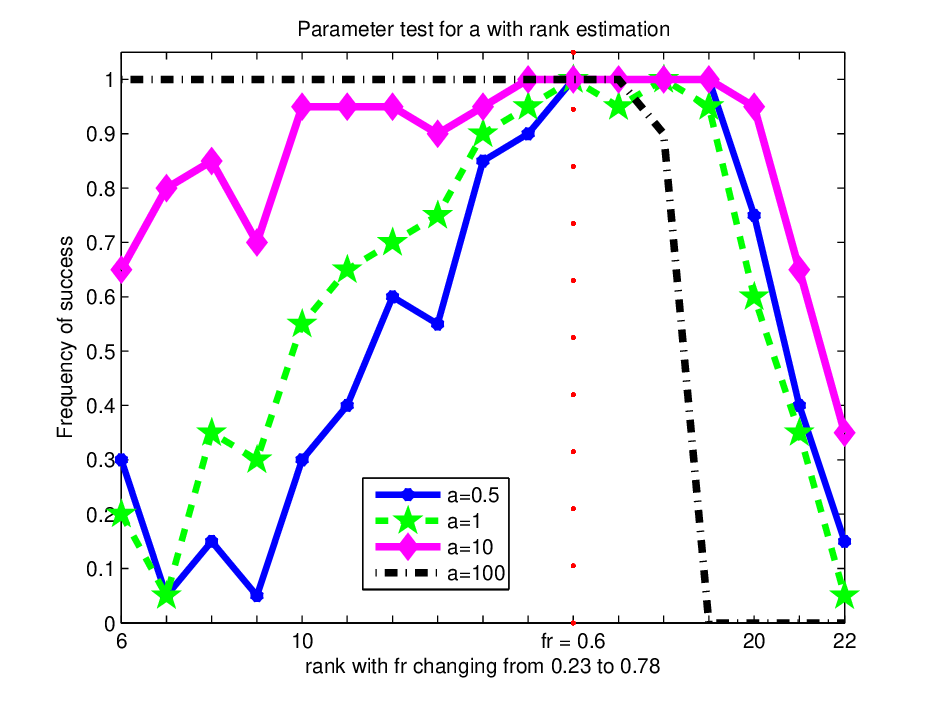}
\caption*{Rank is estimated}   
\end{minipage}
\end{tabular}
\caption{Optimal parameter test for semi-adaptive method: TS1-s1}
\label{figure: optimal a test TS1-s1}
\end{figure}  

It is interesting to see that for known rank scheme, parameter $a =1$ is the optimal strategy, which  
coincides with the optimal parameter setting in \cite{Threshold-TL1}.  It is observed that 
when we use thresholding algorithm under transformed L1 (TL1) or transformed Schatten-1 (TS1) quasi norm, it is usually 
optimal to set $a = 1$ with given information of sparsity or rank. 
However, for the scheme with rank estimation, it is more complicated. Based on our
tests, if FR $< 0.6$, it is better to set $a \geq 100$ to reach good performance. 
On the other hand, if FR $>0.6$, $a= 10$ is nearly the optimal choice. 
So for all the following tests, when we apply TS1-s1 with rank estimation, 
the parameter $a$ is set to be 
\begin{equation*}
a =  \left\{
\begin{array}{ll}
1000, & \ \  \text{if \ FR} < 0.6; \\
10,     & \ \  \text{if \ FR} \geq 0.6.
\end{array}
\right.
\end{equation*} 

In applications where FR is not available, we suggest to use $a=10$, since 
its performance is also acceptable if FR $< 0.6$. 

\subsection{Completion of Random Matrices}
The ground truth matrix $M$ is generated as the matrix product of two low 
rank matrices $M_L$ and $M_R$. Their dimensions are $m \times r$ and 
$n \times r$ respectively, with $r \ll \min(m,n) $. 
In these following experiments, except clearly stated,
$M_L$ and $M_R$ are generated with multivariate normal distribution 
$\mathcal N (\mu, \Sigma)$, with $\mu = 1$ and 
$$
	\Sigma = \{ (1-cov)*I_{(i=j)} + cov  \}_{r \times r} 
$$ 
determined by parameter $cov$. Thus matrix $M = M_L M_R^T$ has rank at most $r$. 

It is known that success recovery is related to FR. The higher FR is, the harder 
it is to recover the original low rank matrix. 
In the first batch of tests, we varied rank $r$ and fixed all other parameters, 
i.e. matrix size $(m,n)$, sampling rate $(sr)$. Thus FR was changing
along with rank. 
 

It is observed that the performance of TS1-s1 and TS1-s2 are very different, 
due to adopting single or double thresholds.  
TS1-s2 uses only one (smooth) thresholding scheme with changing parameter $a$. 
It converges faster than TS1-s1 when the rank is known, see subsection \ref{subsec known rank}. 
On the other hand, TS1-s1 utilizes two (smooth and discontinuous) thresholding schemes, 
and is more robust in case of overestimated rank. TS1-s1 outperforms TS1-s2 when rank 
estimation is used in lieu of the true rank value,  see subsection \ref{subsec unknown rank}. 
IRucL-q method is found to be very robust for varied covariance and rank estimation, 
yet it underperforms TS1 methods at high FR, even with more computing time.  
Though TS1 methods rely on the same rank estimation method as IRucL-q, 
IRucL-q achieves the best results in the absence of true rank value. 
A possible reason is that in IRucL-q iterations, the singular values of matrix 
$X$ are computed more accurately. In TS1, singular values are computed by fast Monte Carlo method 
at every iteration. Due to random 
sampling of Monte Carlo method, there are more errors especially at the beginning stage of iteration. 
The resulting matrices $X^n$ may cause less accurate rank estimation.

\subsubsection{Matrix completion with known rank} 
\label{subsec known rank}

In this subsection, we implemented all six algorithms under 
the condition that true rank value is given. 
They are TS1-s1, TS1-s2, sIRLS-q, IRucL-q, LMaFit and LRGeomCG. 
We skipped FPCA since rank is always adaptively estimated there. 

\subsubsection*{Gaussian matrices with different ranks}
In these tests, matrix $M =  M_L M_R^T$ was generated under uncorrelated 
normal distribution with $\mu =1$. We conducted tests both on low dimensional matrices 
with $m = n = 100$ (Table \ref{table: uncor small known r}) 
and high dimensional matrices with $ m = n = 1000$ (Table \ref{table: uncor large known r}).  
Tests on non-square matrices with $m \neq n$ show similar results. 
  
In Table \ref{table: uncor small known r}, rank $r$ varies from $5$ to $18$, while FR 
increases from $0.2437$ up to $0.8190$.   
For lower rank (less than $15$), LMaFit is the best algorithm with low relative errors 
and fast convergence speed. Part of the reason is that this method does not involve 
SVD (singular value decomposition) operations during iteration.  

LRGeomCG approaches the performance of LMaFit when $r \leq 10$. However, as FR 
values are above $0.7$, it became hard for LMaFit to find truth low rank matrix $M$. 
Its performance is not as good as stated in paper \cite{vandereycken2013low} with
possible reason that we generate M with mean $\mu$ equal to 1, instead of $0$ in 
\cite{vandereycken2013low}. We also tested LRGeomCG with $\mu = 0$ where it has
very small relative error and also fast convergence rate. 

It is also noticed that in Table \ref{table: uncor small known r}, the two TS1 algorithms 
performed very well and remained stable for different FR values. At similar order of 
accuracy, the TL1s are faster than IRucL-q.

\begin{table}
\caption{
Comparison of TS1-s1, TS1-s2, sIRLS-q, IRucL-q, LMaFit and LRGeomCG on recovery of 
uncorrelated multivariate Gaussian matrices at known rank, $m = n = 100 $, SR $=0.4$, 
with stopping criterion $tol = 10^{-6}$.
} 
\label{table: uncor small known r}
\centering
\tabcolsep=0.2cm
\renewcommand*\arraystretch{0.9}
\begin{tabular}{cc|cc|cc|cc}
\hline
\multicolumn{2}{c|}{Problem} & \multicolumn{2}{c|}{TS1-s1}  & \multicolumn{2}{c|}{TS1-s2}  & 
\multicolumn{2}{c|}{sIRLS-q*} \\ \hline
rank & FR & rel.err & time & rel.err & time & rel.err & time  \\ \hline
5 & 0.2437 & 1.89e-05 & 0.11 & 7.58e-07 & 0.13 & 7.09e-06 & 0.80 			    \\ \hline
6 & 0.2910 & 7.13e-06 & 0.14 & 7.37e-07 & 0.15 & 8.59e-06 & 1.01 			\\ \hline
7 & 0.3377 & 1.39e-05 & 0.15 & 6.34e-07 & 0.17 & 8.14e-06 & 1.09 			\\ \hline
8 & 0.3840 & 2.04e-05 & 0.16 & 7.70e-07 & 0.20 & 1.31e-05 & 1.43 				\\ \hline
9 & 0.4298 & 2.08e-05 & 0.23 & 9.97e-07 & 0.25 & 2.02e-05 & 1.88 				\\ \hline
10 & 0.4750 & 3.26e-05 & 0.33 & 1.11e-06 & 0.34 & 1.93e-02 & 4.49 				\\ \hline
14 &	0.6510 &	1.10e-05 & 0.53 & 1.03e-05 & 0.52 & --- & --- 						    \\ \hline
15 &	0.6937 &	1.05e-05 & 0.66 & 9.88e-06 & 0.64 &  --- & --- 					\\ \hline
16 &	0.7360 &	3.86e-05 & 0.91 & 1.79e-05 & 0.87 &  --- & --- 			\\ \hline
17 &	0.7778 &	1.50e-04 & 1.03 & 7.10e-05 & 1.00 &  --- & --- 				\\ \hline
18 &	0.8190 & 	5.63e-04 & 1.00 & 4.15e-04 & 1.00 &  --- & --- 					\\ \hline
\hline  
\multicolumn{2}{c|}{Problem} &  \multicolumn{2}{c|}{IRucL-q} & \multicolumn{2}{c}{LMaFit} 
& \multicolumn{2}{c}{LRGeomCG}   \\ \hline
rank & FR & rel.err & time & rel.err & time & rel.err & time  \\ \hline
5 & 0.2437   &  7.86e-06 & 1.82 & 1.96e-06 & 0.02	& 1.03e-06 & 0.03 			    \\ \hline
6 & 0.2910   &  1.14e-05 & 2.15 & 2.18e-06 & 0.02 	&	1.22e-06 & 0.04  		\\ \hline
7 & 0.3377   &  1.28e-05 & 2.24 & 2.27e-06 & 0.03 	&	1.37e-06 & 0.05		\\ \hline
8 & 0.3840   &  3.03e-05 & 2.33 & 2.67e-06 & 0.03 	&	1.66e-06 & 0.06		\\ \hline
9 & 0.4298   &  1.68e-04 & 2.38 & 3.21e-06 & 0.05 	&	1.88e-06 & 0.07 		\\ \hline
10 & 0.4750  &  3.21e-04 & 2.49 & 3.54e-06 & 0.08	&	1.87e-06 & 0.08		\\ \hline
14 &	0.6510 &   3.80e-05 & 7.25 & 5.74e-06 & 0.21 	&	3.20e-02 & 0.34 				    \\ \hline
15 &	0.6937 &	5.28e-05 & 9.29 & 5.87e-02 & 0.33 	&	3.49e-02 & 0.47				\\ \hline
16 &	0.7360 &	7.57e-05 & 12.34 & 1.44e-01 & 0.34 &	1.91e-01 & 0.99 					\\ \hline
17 &	0.7778 &	9.40e-05 & 15.31 & 3.80e-01 & 0.39 &	 5.73e-01 & 0.71					\\ \hline
18 &	0.8190 & 	1.49e-04 & 22.27 & 4.43e-01 & 0.40 &	 9.17e-01 & 0.94  					\\ \hline

\end{tabular}
\begin{tablenotes}
      \small
      \item * Notes: 1. The sIRLS-q iterations did not converge when rank $> 14$ 
             and FR $\geq 0.65$. Comparison is skipped over this range.
     2. Matrix $M$ is generated from multivariate normal distribution with mean $\mu = 1$,
     instead of $0$.  
    \end{tablenotes}
\end{table}

For large size matrices ($m = n = 1000$), rank $r$ is varied from $50$ to $110$, 
see table \ref{table: uncor large known r}. 
The sIRLS-q and LMaFit only worked for lower FR.  
IRucL-q can still produce satisfactory results with relative error around $10^{-3}$, but its iterations 
took longer time. In \cite{IRucLq}, it was carried out by high speed-performance CPU with many cores. 
Here we used an ordinary processor with only 4 cores and 8 threads.  It is believed that with a better 
machine, IRucL-q will be much faster, since parallel computing is embedded in its codes. 
As seen in the table, LRGeomCG is always convergent and achieves almost same accuracy with TS1-s1 and TS1-s2. 
However, its computation time grows fast with increasing rank. 

A little difference between the two TS1 algorithms began to emerge when matrix size is large.  Although when rank is given, 
they all performed better than other schemes, adaptive TS1-s2 is a little faster than semi-adaptive TS1-s1. It is believed 
by choosing optimal parameter $a$, TS1-s1 will be improved. The parameter $a$ 
is related to matrix $M$, i.e. how it is generated,
its inner structure, and dimension.  
In TS1-s2, the value of parameter $a$ does not need to be manually determined.

\begin{table}
\caption{Numerical experiments on recovery of uncorrelated multivariate Gaussian matrices 
at known rank, $m = n = 1000 $, SR $=0.3$. } 
\label{table: uncor large known r}
\centering
\tabcolsep=0.2cm
\renewcommand*\arraystretch{1}
\begin{tabular}{cc|cc|cc|cc}
\hline
\multicolumn{2}{c|}{Problem} & \multicolumn{2}{c|}{TS1-s1} & \multicolumn{2}{c|}{TS1-s2} 
& \multicolumn{2}{c|}{sIRLS-q}  \\ \hline
rank & FR & rel.err & time & rel.err & time & rel.err & time  \\ \hline
50 &	0.3250  & 5.95e-06 & 8.06 & 5.88e-06 & 6.95 & 4.85e-06 & 45.20 				\\ \hline
70 &	0.4503 & 6.94e-06 & 13.37 & 6.78e-06 & 11.95 & 2.46e-02 & 128.65 			\\ \hline
90 &	0.5730 &	7.83e-06 & 22.13 & 7.77e-06 & 18.81 & 9.86e-02 & 206.32 					\\ \hline
110 &	0.6930 &	1.23e-04 & 29.91 & 3.47e-05 & 29.50 & 2.27e-01 & 282.84  						\\ \hline
\hline
\multicolumn{2}{c|}{Problem} & \multicolumn{2}{c|}{IRucL-q} & \multicolumn{2}{c}{LMaFit}  
& \multicolumn{2}{c}{LRGeomCG} \\ \hline
rank & FR & rel.err & time & rel.err & time & rel.err & time  \\ \hline
50 &	0.3250  & 9.55e-06 & 485.30 & 1.74e-06 & 6.04 	&	1.11e-06 & 8.31  		\\ \hline
70 &	0.4503  & 3.77e-05 & 606.95 & 3.54e-02 & 23.20 	&	1.50e-06 & 20.87		\\ \hline
90 &	0.5730  & 4.16e-04 & 623.37 & 1.60e-01 & 24.94 	&	2.13e-06 & 52.77 	\\ \hline
110 & 0.6930  & 2.41e-03 & 640.66 & 2.45e-01 & 29.19 &	3.22e-06 & 112.30 	\\ \hline 
\end{tabular}
\end{table}

\subsubsection*{Gaussian Matrices with Different Covariance}
In this subsection, the rank $r$, the sampling rate, and the freedom ratio 
FR are fixed. We varied parameter $cov$ to 
generate covariance matrices of multivariate normal distribution. 

In Table \ref{table: cor small known r}, we chose 
two rank values, $r = 5$ and $r = 8$. 
It is harder to recover the original matrix $M$ when it is more 
coherent. IRucL-q does better in this regime. Its mean computing
time and relative errors are less influenced by the changing $cov$.  
Results on large size matrices are shown in Table 
\ref{table: cor large known r}.   
TS1-s2 scheme is much better than TS1-s1, both in 
relative error and computing time. In small size matrix 
experiments, TS1-s2 is the best among comparisons. 

\begin{table}
\caption{ Numerical experiments on multivariate Gaussian matrices with varying covariance at 
known rank, $ m = n = 100 $, SR $=0.4 $. } 
\label{table: cor small known r}
\centering
\tabcolsep=0.2cm
\begin{tabular}{cc|cc|cc|cc}
\hline
\multicolumn{2}{c|}{Problem} & \multicolumn{2}{c|}{TS1-s1} & \multicolumn{2}{c|}{TS1-s2} & 
\multicolumn{2}{c|}{sIRLS-q}  \\ \hline
rank & cor & rel.err & time & rel.err & time & rel.err & time  \\ \hline
5 &	0.5 &	6.44e-06 & 0.17 & 5.74e-07 & 0.12 & 3.35e-02 & 3.75 					\\ \hline
5 &	0.6 &	7.28e-06 & 0.28 & 7.15e-07 & 0.13 & 1.34e-01 & 5.58 						\\ \hline
5 &	0.7 &	3.32e-02 & 0.58 & 7.65e-07 & 0.17 & 2.15e-01 & 6.16 						\\ \hline

8 &	0.4 &	7.55e-06 & 0.34 & 7.96e-07 & 0.21 & 1.43e-01 & 6.47 						\\ \hline
8 &	0.5 &	9.84e-03 & 0.51 & 6.14e-06 & 0.19 & 2.68e-01 & 6.19 						\\ \hline
8 &	0.6 &	3.01e-02 & 0.81 & 7.71e-06 & 0.23 & 2.95e-01 & 6.26  						\\ \hline
8 &	0.7 &	6.86e-02 & 0.86 & 7.16e-06 & 0.50 & 3.33e-01 & 6.80  						\\ \hline

\hline
\multicolumn{2}{c|}{Problem}  & \multicolumn{2}{c|}{IRucL-q} & \multicolumn{2}{c}{LMaFit} 
& \multicolumn{2}{c}{LRGeomCG}   \\ \hline
rank & cor & rel.err & time & rel.err & time & rel.err & time  \\ \hline
5 &	0.5  & 8.21e-06 & 1.86 & 2.48e-02 & 0.07 &	1.12e-06 & 0.06 					\\ \hline
5 &	0.6  & 8.76e-06 & 1.85 & 4.48e-02 & 0.15 &	6.98e-02 & 0.09					\\ \hline
5 &	0.7  & 1.37e-05 & 1.71 & 1.10e-01 & 0.27 &	1.22e-01 & 0.11					\\ \hline 

8 &	0.4  & 1.92e-05 & 2.50 & 1.98e-02 & 0.18 &	5.42e-02 & 0.17 		 \\ \hline
8 &	0.5  & 1.38e-05 & 2.54 & 1.21e-01 & 0.25 &	1.17e-01 & 0.17			  \\ \hline
8 &	0.6  & 1.40e-05 & 2.51 & 1.85e-01 & 0.27 &	1.83e-01 & 0.23			   \\ \hline
8 &	0.7  & 1.10e-05 & 2.35 & 2.44e-01 & 0.25 &	2.21e-01 & 0.29 			\\ \hline
\end{tabular}
\end{table}

In Table \ref{table: cor large known r}, 
we fixed rank $=30$ with $cov$ among $\{ 0.1 , ... , 0. 7\}$. 
TS1-s2 is still satisfactory both in accuracy and speed for low covariance (i.e $cov \leq 0.6 $). 
However, for $cov \geq 0.7$, relative errors increased from $10^{-6}$ to around $10^{-4}$. 
It is also observed that IRucL-q algorithm is very stable and robust 
under covariance change. 

\begin{table}
\caption{Numerical experiments on multivariate Gaussian matrices with varying covariance at known rank, 
$ m = n = 1000 $, SR $=0.4$. } 
\label{table: cor large known r}
\centering
\tabcolsep=0.2cm
\begin{tabular}{cc|cc|cc|cc}
\hline
\multicolumn{2}{c|}{Problem} & \multicolumn{2}{c|}{TS1-s1} & \multicolumn{2}{c|}{TS1-s2} & 
\multicolumn{2}{c|}{sIRLS-q}  \\ \hline
rank & cor & rel.err & time & rel.err & time & rel.err & time   \\ \hline
30 & 	0.1 & 3.07e-06 & 9.71 & 3.07e-06 & 3.98 &   4.36e-07 & 13.80  		\\ \hline		
30 &	0.2 & 2.90e-06 & 11.07 & 2.94e-06 & 3.92 &   1.28e-05 & 33.89 		\\ \hline		
30 &	0.3 & 5.54e-03 & 26.64 & 3.02e-06 & 4.13 &   6.65e-02 & 46.02 		\\ \hline		
30 &	0.4 & 1.19e-02 & 28.58 & 3.08e-06 & 4.31 &   1.08e-01 & 50.95 	\\ \hline		
30 &	0.5 & 4.76e-02 & 34.25 & 2.89e-06 & 5.89 &   1.50e-01 & 52.64  	\\ \hline		
30 &	0.6 & 6.89e-02 & 35.69 & 2.89e-06 & 10.28 &   1.89e-01 & 55.70 	\\ \hline		
30 &	0.7 & 8.01e-02 & 33.92 & 6.99e-04 & 20.09 &   2.03e-01 & 51.03  	\\ \hline	

\hline
\multicolumn{2}{c|}{Problem}  & \multicolumn{2}{c|}{IRucL-q} & \multicolumn{2}{c}{LMaFit} 
& \multicolumn{2}{c}{LRGeomCG} \\ \hline
rank & cor & rel.err & time & rel.err & time & rel.err & time   \\ \hline
30 & 	0.1  & 3.13e-06 & 222.90 & 1.19e-06 & 1.83 & 	6.77e-07 & 4.88 	\\ \hline		
30 &	0.2  & 3.16e-06 & 221.34 & 1.14e-06 & 3.16 &	5.68e-07 & 8.84	\\ \hline		
30 &	0.3  & 3.05e-06 & 218.57 & 1.21e-06 & 6.93 &	5.45e-03 & 15.45 	\\ \hline		
30 &	0.4  & 3.29e-06 & 214.52 & 2.06e-02 & 14.72 & 4.82e-02 & 19.15	\\ \hline		
30 &	0.5  & 3.12e-06 & 209.05 & 6.45e-02 & 17.34 & 8.41e-02 & 20.99	\\ \hline		
30 &	0.6  & 3.30e-06 & 207.94 & 9.09e-02 & 18.38 & 1.42e-01 & 21.81	\\ \hline		
30 &	0.7  & 3.15e-06 & 210.06 & 1.15e-01 & 16.37 & 1.67e-01 & 21.63 	\\ \hline	

\end{tabular}
\end{table}

\subsubsection*{Matrices from other distributions}
We also compare algorithms with other distributions, including $(0,1)$ uniform 
distribution and Chi-square distribution with k = 1 (degree of freedom). All other 
parameters are same as Table \ref{table: uncor small known r}. 
The results are displayed at Table \ref{table: uniform small known r} (uniform 
distribution) and Table \ref{table: Chi-square small known r} (Chi-square distribution).
Only partial numerical results are showed here with rank $r = 7,8,9,10,14,15.$ 
From these two tables, two TS1 algorithms have satisfying relative errors and stable 
performance, same as IRuccL-q. For these two non-Gaussian distributions, it becomes 
harder to successfully recover low rank matrix for LMaFit  and LRGeomCG, especially 
when rank $r>10$.

\begin{table}
\caption{
Comparison with random matrices generated from $(0,1)$ 
uniform distribution. Rank $r$ is given and $m = n = 100 $, 
SR $=0.4$, with stopping criterion $tol = 10^{-6}$.
} 
\label{table: uniform small known r}
\centering
\tabcolsep=0.2cm
\renewcommand*\arraystretch{0.9}
\begin{tabular}{cc|cc|cc|cc}
\hline
\multicolumn{2}{c|}{Problem} & \multicolumn{2}{c|}{TS1-s1}  & \multicolumn{2}{c|}{TS1-s2}  & 
\multicolumn{2}{c|}{sIRLS-q*}  
\\ \hline
rank & FR & rel.err & time & rel.err & time & rel.err & time  \\ \hline
7  & 0.3377 &  5.67e-06 & 0.16 & 5.30e-06 & 0.14 & 7.30e-06 & 1.85   \\ \hline
8  & 0.3840 &  6.73e-06 & 0.18 & 6.46e-06 & 0.15 & 1.96e-02 & 3.78   \\ \hline
9  & 0.4298 &  9.13e-06 & 0.24 & 8.42e-06 & 0.20 & ---      & ---    \\ \hline
10 & 0.4750 &  7.62e-06 & 0.27 & 7.12e-06 & 0.20 & ---      & ---    \\ \hline
14 & 0.6510 &  2.23e-05 & 0.59 & 9.24e-06 & 0.44 & ---      & ---    \\ \hline
15 & 0.6937 &  2.34e-05 & 0.81 & 1.12e-05 & 0.58 & ---      & ---    \\ \hline
\hline  
\multicolumn{2}{c|}{Problem} &  \multicolumn{2}{c|}{IRucL-q} & \multicolumn{2}{c}{LMaFit} 
& \multicolumn{2}{c}{LRGeomCG}  \\ \hline
rank & FR & rel.err & time & rel.err & time & rel.err & time  \\ \hline
7  & 0.3377 &  9.55e-06 & 5.00  & 1.98e-06 & 0.05 &  1.48e-06 & 0.08  \\ \hline
8  & 0.3840 &  1.08e-05 & 4.86  & 2.41e-06 & 0.06 &  1.58e-06 & 0.10  \\ \hline
9  & 0.4298 &  1.57e-05 & 6.48  & 2.26e-02 & 0.13 &  2.01e-06 & 0.14  \\ \hline
10 & 0.4750 &  1.80e-05 & 7.09  & 7.28e-03 & 0.11 &  2.09e-06 & 0.13  \\ \hline
14 & 0.6510 &  3.75e-05 & 13.15 & 1.66e-01 & 0.18 &  1.24e-01 & 0.44  \\ \hline
15 & 0.6937 &  5.58e-05 & 17.14 & 2.18e-01 & 0.16 &  1.71e-01 & 0.76  \\ \hline
\end{tabular}
\end{table}

\begin{table}
\caption{
Comparison with random matrices generated from Chi-square distribution with k = 1 
(degree of freedom).
Rank $r$ is given and $m = n = 100 $, 
SR $=0.4$, with stopping criterion $tol = 10^{-6}$.
} 
\label{table: Chi-square small known r}
\centering
\tabcolsep=0.2cm
\renewcommand*\arraystretch{0.9}
\begin{tabular}{cc|cc|cc|cc}
\hline
\multicolumn{2}{c|}{Problem} & \multicolumn{2}{c|}{TS1-s1}  & \multicolumn{2}{c|}{TS1-s2}  & 
\multicolumn{2}{c|}{sIRLS-q*}  
\\ \hline
rank & FR & rel.err & time & rel.err & time & rel.err & time  \\ \hline
7  & 0.3377 & 9.09e-06 & 0.23  & 8.56e-06 & 0.20 & 1.82e-05 & 1.84 \\ \hline
8  & 0.3840 & 1.06e-05 & 0.27  & 8.31e-06 & 0.22 & 1.69e-02 & 2.59 \\ \hline
9  & 0.4298 & 9.90e-06 & 0.30  & 8.79e-06 & 0.25 & ---      & ---    \\ \hline
10 & 0.4750 & 9.52e-06 & 0.33  & 8.64e-06 & 0.28 & ---      & ---    \\ \hline
14 & 0.6510 & 1.48e-05 & 0.64  & 1.20e-05 & 0.58 & ---      & ---    \\ \hline
15 & 0.6937 & 2.23e-05 & 0.83  & 1.32e-05 & 0.73 & ---      & ---    \\ \hline
\multicolumn{2}{c|}{Problem} &  \multicolumn{2}{c|}{IRucL-q} & \multicolumn{2}{c}{LMaFit} 
& \multicolumn{2}{c}{LRGeomCG}  \\ \hline
rank & FR & rel.err & time & rel.err & time & rel.err & time  \\ \hline
7  & 0.3377 & 1.26e-05 & 5.65  & 3.08e-06 & 0.04 & 1.80e-06 & 0.05 \\ \hline
8  & 0.3840 & 1.70e-05 & 7.15  & 3.29e-06 & 0.04 & 2.19e-06 & 0.06 \\ \hline
9  & 0.4298 & 2.21e-05 & 8.33  & 3.75e-06 & 0.08 & 6.83e-03 & 0.11 \\ \hline
10 & 0.4750 & 2.23e-05 & 8.56  & 4.25e-06 & 0.09 & 5.93e-02 & 0.14 \\ \hline
14 & 0.6510 & 5.50e-05 & 14.69 & 1.44e-01 & 0.15 & 1.46e-01 & 0.34 \\ \hline
15 & 0.6937 & 6.61e-05 & 17.75 & 2.54e-01 & 0.15 & 3.03e-01 & 0.57 \\ \hline
\end{tabular}
\end{table}

\subsubsection{Matrix completion with rank estimation}
\label{subsec unknown rank}

We conducted numerical experiments on rank estimation schemes. 
The initial rank estimation is given 
as $1.5\, r$, which is a commonly used overestimate. 
FPCA \cite{FPCA} is included for comparison, while LRGeomCG and sIRLS-q are excluded.
FPCA is a fast and robust iterative algorithm based on nuclear norm regularization. 

We considered two classes of matrices: uncorrelated Gaussian matrices with changing rank; 
correlated Gaussian matrices with fixed rank ($r=5, 10$). 
The results are shown in Table \ref{table: uncor small unknown r} 
and Table \ref{table: cor small unknown r}. 
It is interesting that under rank estimation,
the semi-adaptive TS1-s1 fared much better than TS1-s2.  
In low rank and low covariance cases, TS1-s1 is the best in terms of 
accuracy and computing time among  
comparisons. However, in the regime of high covariance and rank, it became harder for TS1 
methods to perform efficient recovery. IRucL-q did the best, being 
both stable and robust. In the most difficult case,  
at $rank = 15$ and FR approximately equal to $0.7$, IRucL-q can still obtain an accurate 
result with relative error around $10^{-5}$.

\begin{table}
\caption{Numerical experiments for low rank matrix completion algorithms under rank estimation.  
True matrices are uncorrelated multivariate Gaussian,  $ m = n = 100 $, SR $=0.4$. } 
\label{table: uncor small unknown r}
\centering
\tabcolsep=0.05cm
\begin{tabular}{cc|cc|cc|cc|cc|cc|cc}
\hline
\multicolumn{2}{c|}{Problem} & \multicolumn{2}{c|}{TS1-s1} 
& \multicolumn{2}{c|}{TS1-s2} & \multicolumn{2}{c|}{FPCA} 
& \multicolumn{2}{c|}{IRucL-q} & \multicolumn{2}{c}{LMaFit} \\ \hline
rank & FR & rel.err & time & rel.err & time & rel.err & time 
& rel.err & time & rel.err & time \\ \hline
10 &	0.4750 &	7.46e-06 & 0.31 &  2.43e-03 & 0.38 & 2.26e-01 & 0.91 
& 1.84e-05 & 3.41 & 2.64e-01 & 0.01 						\\ \hline
11 &	0.5198 &	1.04e-05 & 0.35 &  1.15e-02 & 0.52 & 2.23e-01 & 0.88 
& 2.15e-05 & 4.09 & 2.48e-01 & 0.01 						\\ \hline
12 &	0.5640 &	9.94e-06 & 0.44 &  7.62e-03 & 0.54 & 2.28e-01 & 0.92 
& 2.51e-05 & 4.46 & 2.44e-01 & 0.01 						\\ \hline
13 &	0.6078 &	3.71e-02 & 0.80 &  5.71e-03 & 0.68 & 2.25e-01 & 0.84 
& 3.35e-05 & 5.61 & 2.24e-01 & 0.02 						\\ \hline
14 &	0.6510 &	7.02e-03 & 0.82 &  1.03e-03 & 0.65 & 2.23e-01 & 0.88 
& 3.97e-05 & 6.41 & 2.19e-01 & 0.01 						\\ \hline
15 &	0.6937 &	4.96e-03 & 0.95 &  2.88e-03 & 0.92 & 2.18e-01 & 0.88 
& 4.82e-05 & 7.86 & 2.12e-01 & 0.02 						\\ \hline
\end{tabular}
\end{table}

\begin{table}
\caption{Numerical experiments on low rank matrix completion algorithms under rank estimation. 
True matrices are multivariate Gaussian with different covariance, $m = n = 100 $, and SR $=0.4$.} 
\label{table: cor small unknown r}
\centering
\tabcolsep=0.07cm
\begin{tabular}{cc|cc|cc|cc|cc|cc|cc}
\hline
\multicolumn{2}{c|}{Problem} & \multicolumn{2}{c|}{TS1-s1} 
& \multicolumn{2}{c|}{TS1-s2} & \multicolumn{2}{c|}{FPCA}  
& \multicolumn{2}{c|}{IRucL-q} & \multicolumn{2}{c}{LMaFit} \\ \hline
rank & cor & rel.err & time & rel.err & time & rel.err & time 
& rel.err & time & rel.err & time \\ \hline
5 &	0.5 &	5.49e-06 & 0.20 &  6.77e-02 & 0.86 & 1.61e-05 & 0.12 
& 7.50e-06 & 2.07 & 1.24e-01 & 0.01						\\ \hline
5 &	0.6 &	5.45e-06 & 0.20 &  7.74e-02 & 0.91 & 1.69e-05 & 0.11 
& 6.93e-06 & 1.76 & 9.12e-02 & 0.01 						\\ \hline
5 &	0.7 &	5.25e-06 & 0.25 &  1.04e-01 & 1.33 & 1.53e-05 & 0.12 
& 4.71e-04 & 2.06 & 6.60e-02 & 0.01						\\ \hline  \hline
10 &	0.5 &	1.10e-05 & 0.65 &  1.17e-01 & 1.14 & 1.21e-01 & 0.97 
& 1.76e-05 & 3.35 & 9.66e-02 & 0.01 						\\ \hline
10 &	0.6 &	1.61e-02 & 0.76 &  1.32e-01 & 1.04 & 1.02e-01 & 0.86 
& 2.72e-05 & 4.26 & 7.33e-02 & 0.01 						\\ \hline
10 &	0.7 &	9.14e-02 & 0.91 &  1.55e-01 & 0.93 & 9.11e-02 & 0.82 
& 7.12e-04 & 4.59 & 5.06e-02 & 0.01 						\\ \hline
\end{tabular}
\end{table}

\subsection{Image inpainting}

As in \cite{IRucLq,LMaFit}, we conducted grayscale image inpainting experiments 
to recover low rank images from partial observations, and compare with IRcuL-q 
and LMaFit algorithms.  
The `boat' image (see Figure \ref{figure: boat image}) is used to produce ground 
truth as in \cite{IRucLq} with rank equal to $40$ and at $512 \times 512$ resolution. 
Different levels of noisy disturbances are added to the original image $M_o$ by the formula
\begin{equation*}
M = M_o + \sigma \frac{\| M_o \|_F}{\| \varepsilon \|_F} \varepsilon,
\end{equation*}
where the matrix $\varepsilon$ is a standard Gaussian. 

\begin{figure}
\begin{tabular}{lll}
\begin{minipage}[t]{0.31 \linewidth}
\centering 
\includegraphics[scale=0.3]{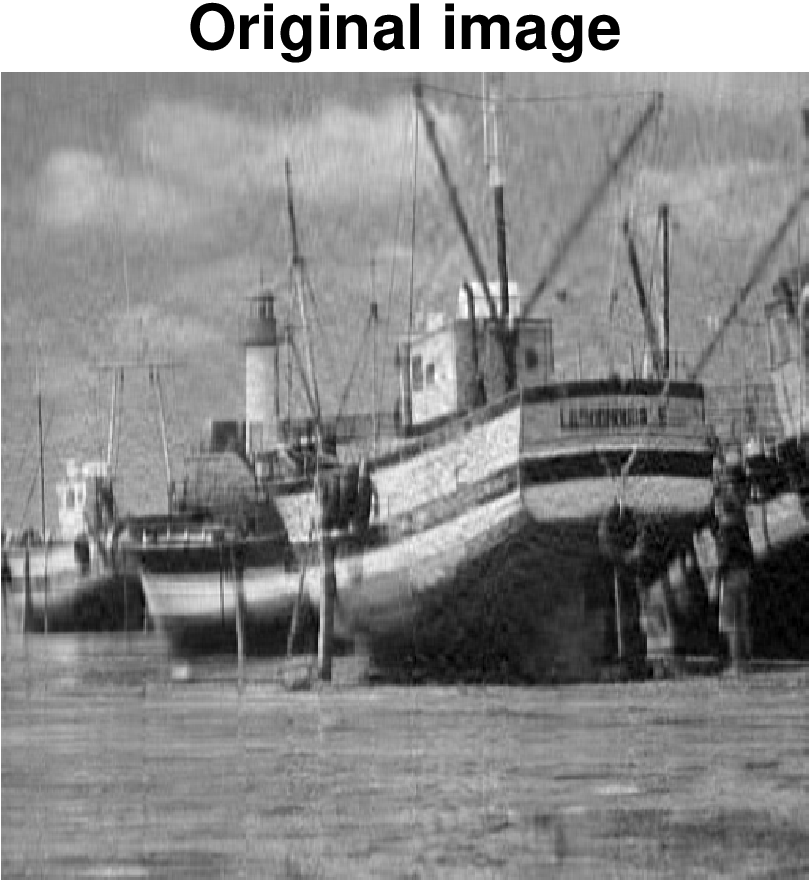}  
\end{minipage}  &
\begin{minipage}[t]{0.31 \linewidth}
\centering 
\includegraphics[scale=0.3]{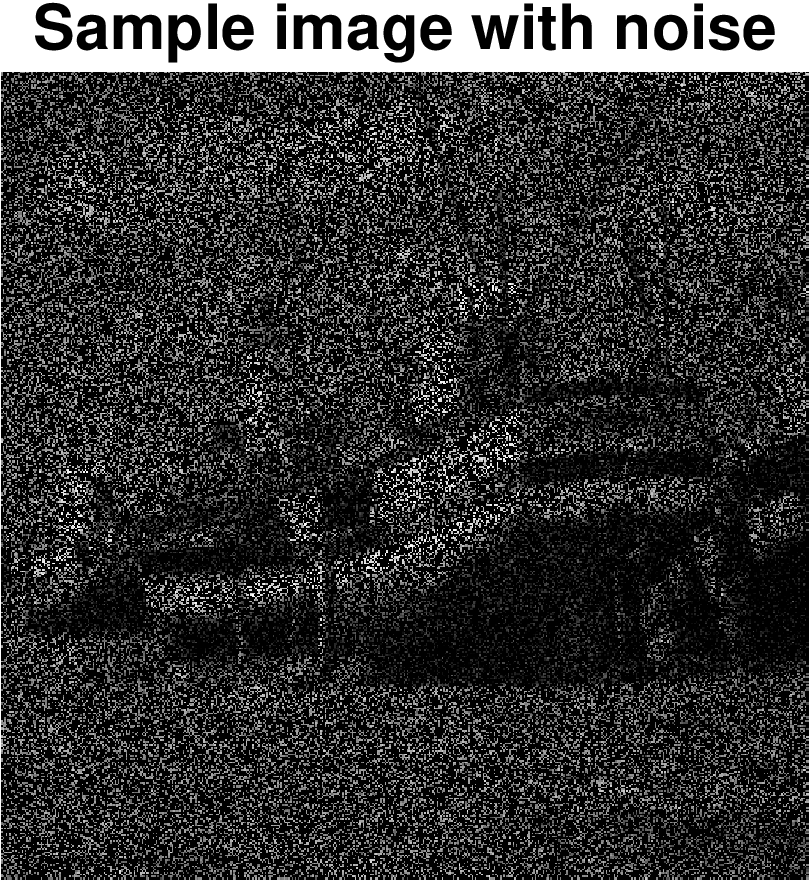}
\end{minipage} &
\begin{minipage}[t]{0.31 \linewidth}
\centering 
\includegraphics[scale=0.3]{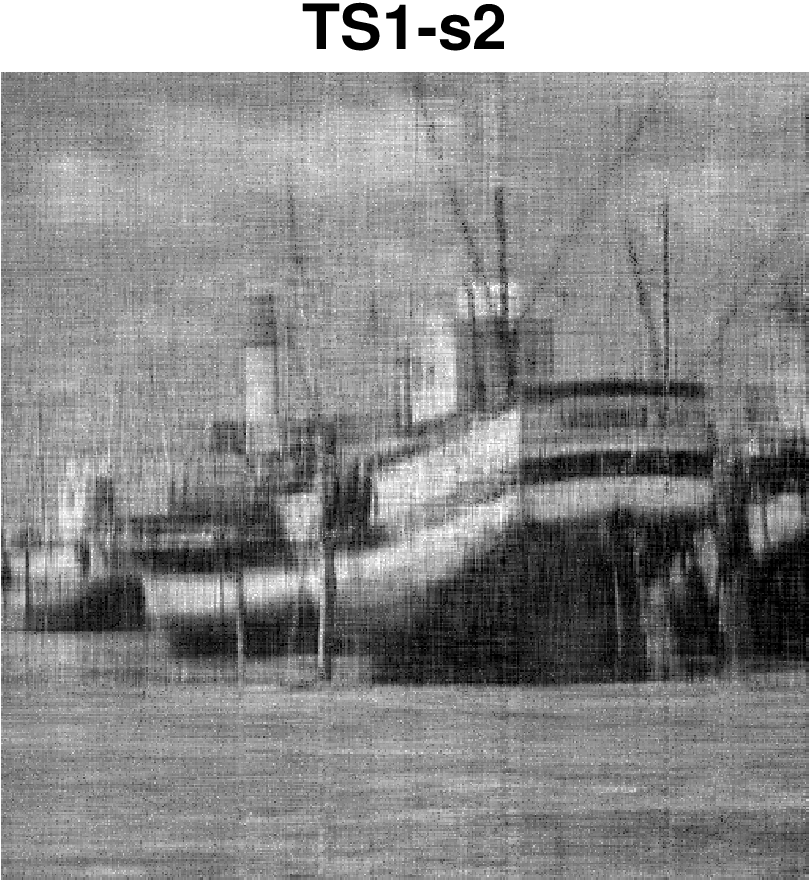}
\end{minipage} \\
\begin{minipage}[t]{0.31 \linewidth}
\centering 
\includegraphics[scale=0.3]{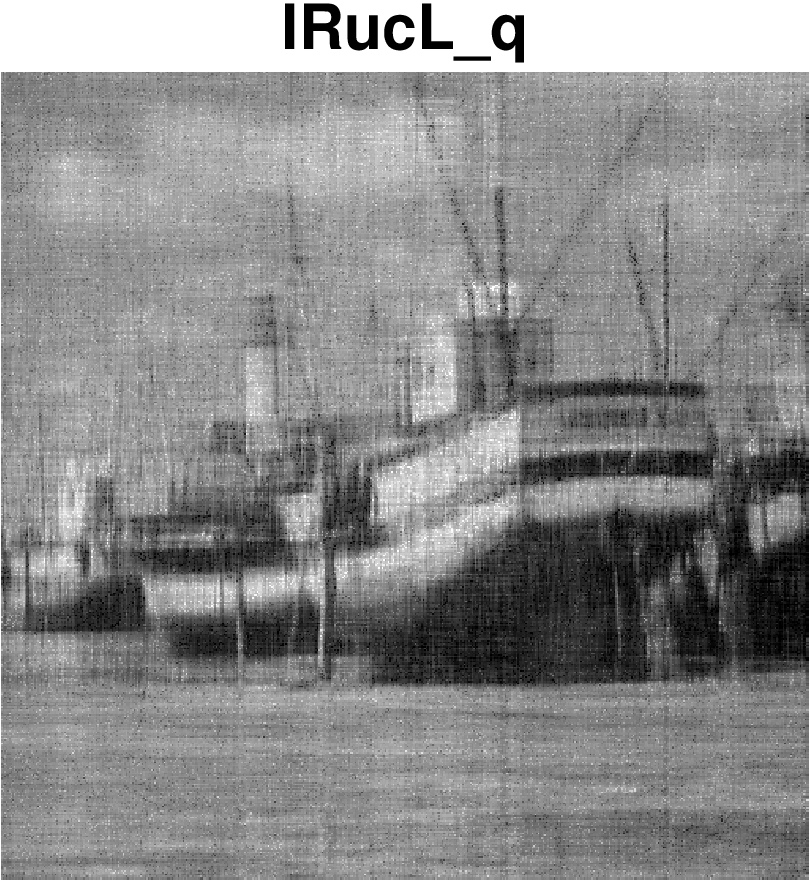} 
\end{minipage}  &
\begin{minipage}[t]{0.31 \linewidth}
\centering 
\includegraphics[scale=0.3]{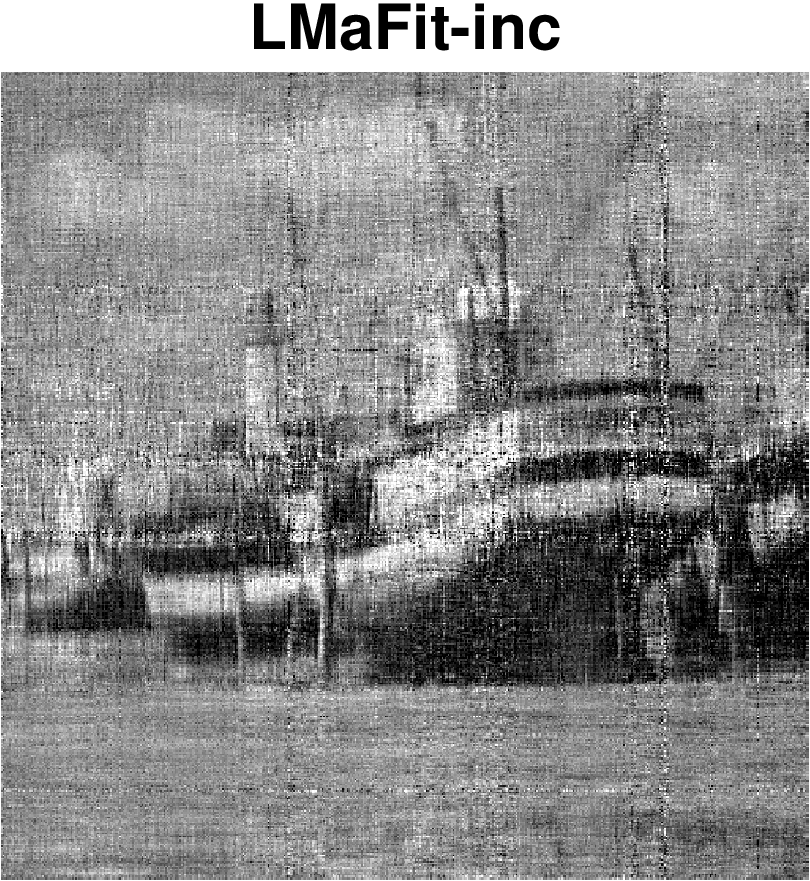}
\end{minipage} &
\begin{minipage}[t]{0.31 \linewidth}
\centering 
\includegraphics[scale=0.3]{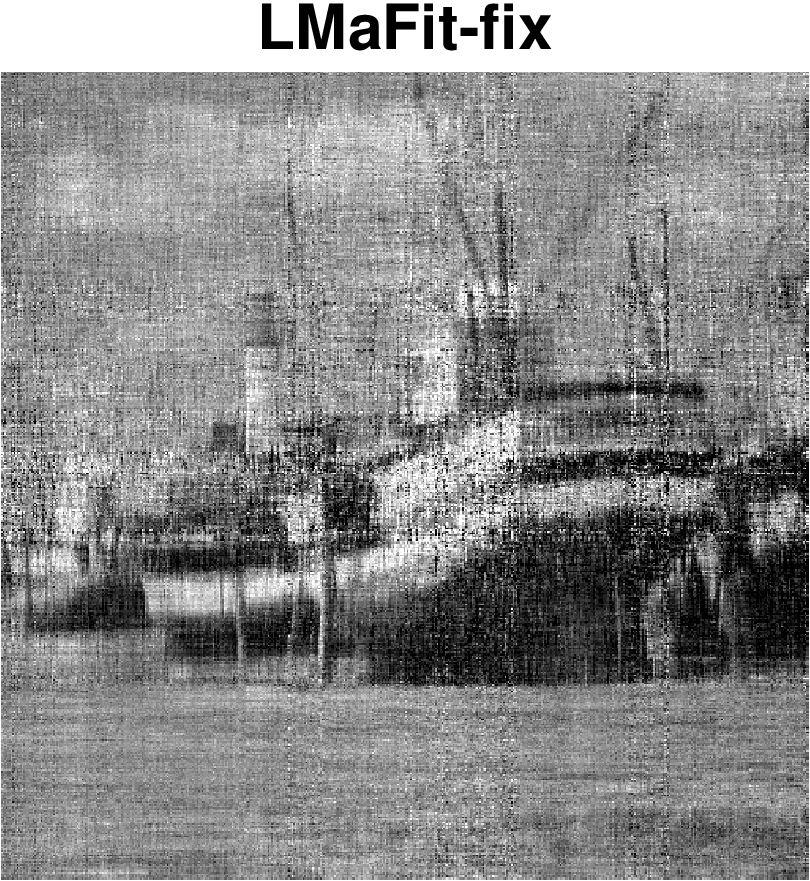}
\end{minipage}
\end{tabular}
\caption{
Image inpainting experiments with 
$\mathrm{SR} = 0.3, \sigma = 0.15$. }
\label{figure: boat image}
\end{figure}  

Here we only applied scheme TS1-s2. For IRucL-q, we followed the setting in \cite{IRucLq}  
by choosing $\alpha = 0.9$ and $\lambda = 10^{-2} \sigma$. 
Both fixed rank ( LMaFit-fix ) and increased rank (LMaFit-inc) schemes are implemented 
for LMaFit. We took fixed rank $r=40$ for TS1-s2, LMaFit-fix and IRucL-q. 

Computational results are in Table \ref{table: image} with sampling ratios varying among 
$\{ 0.3, 0.4, 0.5\}$ and noise strength $\sigma$ in $\{ 0.01, 0.05, 0.10, 0.15, 0.20, 0.25 \}$. 
The performance for each algorithm is measured in CPU time, PSNR (peak-signal noise ratio), and 
MSE (mean squared error). Here we focus more on PSNR values and placed the top $2$ 
in bold for each experiment. We observed that IRucL-q and TS1-s2 fared about the same. Either one is 
better than LMaFit in most cases. 

\begin{table}
\caption{Numerical experiments on boat image inpainting with algorithms TS1, IRcuL-q and LMaFit under 
different sampling ratio and noise levels.} 
\label{table: image}
\centering
\tabcolsep=0.03cm
\begin{tabular}{cc|ccc|ccc|ccc|lll}
\hline
\multicolumn{2}{c|}{Problem} & \multicolumn{3}{c|}{TS1-s2} & \multicolumn{3}{c|}{IRucL-q} 
& \multicolumn{3}{c|}{LMaFit-inc} & \multicolumn{3}{c|}{LMaFit-fix}  \\ \hline
SR & $\sigma$  & Time & PSNR & MSE  & Time & PSNR & MSE  & Time & PSNR & MSE  & Time & PSNR & MSE  \\ \hline
0.3 &	0.01 &	  27.23 & \textbf{44.21} & 3.79e-5&    85.97 & 43.28 & 4.70e-5&     5.70 & 32.80 & 5.25e-4&     2.17 & \textbf{45.02} & 3.15e-5 								\\
0.3 &	0.05 &	  27.81 & \textbf{30.55} & 8.82e-4&    58.25 & \textbf{29.55} & 1.11e-3&     6.00 & 29.10 & 1.23e-3&     2.81 & 29.28 & 1.18e-3 								\\
0.3 &	0.10 &	  29.21 & \textbf{24.89} & 3.24e-3&    24.26 & \textbf{24.99} & 3.17e-3&     5.59 & 19.74 & 1.06e-2&     5.74 & 18.52 & 1.41e-2 								\\
0.3 &	0.15 &	  26.37 & \textbf{22.57} & 5.54e-3&    27.61 & \textbf{22.74} & 5.33e-3&     5.46 & 16.64 & 2.17e-2&     4.84 & 15.98 & 2.52e-2 								\\
0.3 &	0.20 & 	  26.75 & \textbf{20.89} & 8.14e-3&    24.45 & \textbf{21.05} & 7.85e-3&     5.95 & 14.68 & 3.41e-2&     3.52 & 14.03 & 3.95e-2 								\\
0.3 &	0.25 & 	  26.92 & \textbf{19.60} & 1.10e-2&    23.75 & \textbf{19.75} & 1.06e-2&     5.52 & 12.91 & 5.12e-2&     1.85 & 12.73 & 5.33e-2 								\\		\hline

0.4 & 	0.01 &	   26.29 & 44.30 & 3.71e-5&    80.19 & 43.25 & 4.74e-5&     6.53 & \textbf{44.84} & 3.28e-5&     2.93 & \textbf{45.02} & 3.15e-5    \\								
0.4 & 	0.05 &	   26.05 & \textbf{30.58} & 8.75e-4&     63.20 & \textbf{29.39} & 1.15e-3&     4.62 & 29.09 & 1.23e-3&     3.12 & 27.91 & 1.62e-3 		\\						
0.4 & 	0.10 &	   26.08 & \textbf{24.74} & 3.35e-3&    32.58 & \textbf{24.86} & 3.27e-3&     6.44 & 19.97 & 1.01e-2&     8.00 & 19.19 & 1.21e-2 			\\					
0.4 & 	0.15 &	   26.34 & \textbf{22.57} & 5.53e-3&    26.30 & \textbf{22.72} & 5.35e-3&     5.52 & 16.78 & 2.10e-2&     2.86 & 16.21 & 2.40e-2 				\\				
0.4 & 	0.20 & 	   29.04 & \textbf{20.89} & 8.15e-3&    20.73 & \textbf{21.08} & 7.81e-3&     5.44 & 14.47 & 3.58e-2&     2.25 & 14.43 & 3.61e-2 					\\			
0.4 & 	0.25 & 	   28.84 & \textbf{19.56} & 1.11e-2&    20.48 & \textbf{19.68} & 1.08e-2&     5.70 & 12.79 & 5.26e-2&     2.35 & 12.57 & 5.54e-2 						\\			\hline

0.5 & 	0.01 &	   27.76 & \textbf{44.26} & 3.75e-5&    82.42 & 43.30 & 4.67e-5&     5.04 & 34.50 & 3.55e-4&     2.79 & \textbf{45.01} & 3.15e-5  \\								
0.5 & 	0.05 &	   27.89 & \textbf{30.54} & 8.82e-4&    64.19 & 29.47 & 1.13e-3&     5.81 & 28.63 & 1.37e-3&     2.79 & \textbf{29.62} & 1.09e-3 		\\						
0.5 & 	0.10 &	   29.56 & \textbf{24.80} & 3.31e-3&    30.50 & \textbf{24.94} & 3.21e-3&     5.78 & 19.92 & 1.02e-2&     3.54 & 19.09 & 1.23e-2 			\\					
0.5 & 	0.15 &	   26.21 & \textbf{22.59} & 5.51e-3&    24.24 & \textbf{22.74} & 5.32e-3&     5.71 & 16.73 & 2.12e-2&     2.67 & 16.32 & 2.33e-2 				\\				
0.5 & 	0.20 & 	   28.01 & \textbf{20.89} & 8.14e-3&    22.51 & \textbf{21.07} & 7.82e-3&     4.44 & 15.67 & 2.71e-2&     2.42 & 14.38 & 3.65e-2 					\\			
0.5 & 	0.25 & 	   29.86 & \textbf{19.52} & 1.12e-2&    18.32 & \textbf{19.71} & 1.07e-2&     5.54 & 12.62 & 5.48e-2&     3.24 & 12.74 & 5.32e-2 						\\   \hline		
\end{tabular}
\end{table}

\section{Conclusion}   \label{section: conclusion}
We presented the transformed Schatten-1 penalty (TS1), and derived the closed 
form thresholding representation formula for global minimizers of TS1 regularized rank 
minimization problem. We studied two adaptive iterative TS1 schemes (TS1-s1 and TS1-s2) 
computationally for matrix completion in comparison with several state-of-art methods, 
in particular \IRq . 
In case of low rank matrix recovery under known rank, TS1-s2 performs the best 
in accuracy and computational speed. In low rank matrix recovery under rank estimation, 
TS1-s1 is almost on par with IRucL-q except when both the matrix covariance and rank 
rise to certain level. In future work, we shall study rank estimation techniques to further 
improve on TS1-s1 and explore other applications for TS1 penalty. 

\section*{Acknowledgments}
The authors would like to thank Prof. Wotao Yin for his helpful suggestions on low rank matrix 
completion methods and numerical experiments. 

\appendix[Proof of Ky Fan k-norm inequality]
\begin{proof} 
Since $X = U \mathrm{Diag}(\sigma) V^T$, the $(j,k)$-th entry of matrix $X$ is 
$X_{j,k} = \sum \limits_{i = 1}^{m} \sigma_i U_{j,i} V_{k,i}$.

Thus, we have
\begin{equation} \label{equ: trk(X)}
\begin{array}{lll}
\mathrm{tr_k}(X) & = & \sum \limits_{j = 1}^{k} X_{j,j} 
           =  \sum \limits_{j = 1}^{k} \sum \limits_{i = 1}^{m} \sigma_i U_{j,i} V_{j,i}  \\
        & = & \sum \limits_{i = 1}^{m} \sum \limits_{j = 1}^{k} \sigma_i U_{j,i} V_{j,i} 
           =  \sum \limits_{i = 1}^{m} \sigma_i w_i^{(k)},        
\end{array}
\end{equation}
where the weight $w_i^{(k)}$ for the singular value $\sigma_i$ is defined as: 
\begin{equation} \label{func: weight w_i^k}
w_i^{(k)} = \sum \limits_{j = 1}^{k} U_{j,i} V_{j,i},  \ \ i = 1,2,...,m .
\end{equation}

Notice that, 
\begin{equation} \label{inequ: wik 1}
|w_i^{(k)}| \leq \sum \limits_{j = 1}^{k} |U_{j,i}| |V_{j,i}| \leq \|U(:,i)\|_2 \|V(:,i)\|_2 \leq 1,
\end{equation} 
where $U(:,i)$ and $V(:,i)$ are the $i$-th column vectors for $U$ and $V$.
Also for weights $\{w_i^{(k)}\}$,
\begin{equation} \label{inequ: wik sum k}
\begin{array}{lll}
\sum \limits_{i = 1}^{m} |w_i^{(k)}| 
	& \leq & \sum \limits_{i = 1}^{m} \sum \limits_{j = 1}^{k} |U_{j,i}||V_{j,i}|
       =  \sum \limits_{j = 1}^{k} \sum \limits_{i = 1}^{m} |U_{j,i}||V_{j,i}| \\  
    & \leq & \sum \limits_{j = 1}^{k} \ \|U(j,:)\|_2 \ \|V(j,:)\|_2 \leq  k,
\end{array}
\end{equation} 
where $U(j,:)$ and $V(j,:)$ are the $j$-th row vectors for $U$ and $V$, respectively.

All the $m$ weights are bounded by $1$, with absolute sum at most $k \leq m$. 
Note that $\sigma_i$'s are in decreasing order. 
By equation (\ref{equ: trk(X)}), we have, for all $k = 1,2,...,m$,
$$
\trk(X) \leq \sum \limits_{i = 1}^{m} \sigma_i |w_i^{(k)}|  \leq 
\sum \limits_{i = 1}^{k} \sigma_i = \trk(D) =\|X\|_{Fk}.
$$ 
 
Next, we prove the second part of the lemma --- equality condition, by mathematical induction. 
Suppose that for a given matrix $X$,
$\trk(X) = \trk(D)$, $\forall \ k = 1,...,m$. 
Here, it is convenient to define $X_i = \sigma_i U_i V_i^T$,
where $V_i$ ($U_i$) is the $i$-th column vector of $V$ ($U$).
Then matrix $X$ can be decomposed as the sum of $r$ rank-$1$ matrices, 
$X = \sum \limits_{i = 1}^{r} X_i$. 

When $k=1$, according to $tr_{1}(X) = tr_{1}(D)$ and the proof above, we know that
$$
w^{(1)}_1 = 1 \ \ \text{and}  \ w^{(1)}_i = 0 \ \ \text{for} \ i = 2, ... , m.
$$
By the definition of weights $w_i^{(k)}$ in (\ref{func: weight w_i^k}), 
we have $w_1^{(1)} = U_{1,1} V_{1,1} = 1$. 
Since $U$ and $V$ are both unitary matrices, we have:
$$
U_{1,1} = V_{1,1} = \pm 1; \ \ \ U_{1,j} = U_{j,1} = V_{1,j} = V_{j,1} = 0 \ \text{for} \ j \neq 1.
$$
Then vectors $U_1$ ($V_1$) is the first standard basis vector in space $\Re^m$ ($\Re^n$). 
The matrix $X_1 = \sigma_1 U_1 V_1^T$ is diagonal
$$ 
X_1 = \left[ \begin{array}{cccc}
\sigma_1 &    &        &   \\
         & 0  &        &   \\
         &    & \ddots &   \\
         &    &        & 0
         \end{array} \right]_{m \times n}
$$ 

For any index $i$, $1 \leq i \leq k-1$, suppose that 
\begin{equation} \label{equ: Ui,j Vi,j assum cond}
U_{i,i} = V_{i,i} = \pm 1; \ \ \ U_{i,j} = U_{j,i} = V_{i,j} = V_{j,i} = 0 \ \text{for any index} \ j \neq i.
\end{equation}
Then matrix $X_i = \sigma_i U_i V_i^T$, with $1 \leq i \leq k-1$, 
is diagonal and can be expressed as 
$$ 
X_i = \left[ \begin{array}{ccccccc}
 0       &         &          &                                                     \\
         & \ddots  &          &                                                     \\
         &         & 0        &                                                     \\
         &         &          &  \sigma_i   &                                       \\
         &         &          &             &  0                                    \\
         &         &          &             &           & \ddots      &              \\ 
         &         &          &             &           &             &  0          \\     
         \end{array} \right]_{m \times n} \longleftarrow (i \text{-th row})
$$ 

Under those conditions, let us consider the case with index $i = k$. Clearly, we have
$tr_{k}(X) = tr_{k}(D)$. Similarly as before, thanks to the formula (\ref{equ: trk(X)}) and 
inequalities (\ref{inequ: wik 1}) and (\ref{inequ: wik sum k}), it is true that
$$
w_i^{(k)} = 1 \ \text{for} \  i = 1, ... ,k; \ \ \ \text{and} \ 
w_i^{(k)} = 0 \ \text{for} \  i > k.
$$

Furthermore, by definition (\ref{func: weight w_i^k}), 
$w_k^{(k)} = \sum \limits_{j = 1}^{k} U_{j,k} V_{j,k} = U_{k,k} V_{k,k} = 1$.  This is 
because $U_{j,k} = V_{j,k} = 0$ for index $j < k$, by the 
assumption (\ref{equ: Ui,j Vi,j assum cond}) .
Thus vectors $U_k$ and $V_k$ are also standard basis vectors with the $k$-th entry 
to be $\pm 1$. Then
$$ 
X_k = \sigma_k U_k V_k^T = 
\left[ \begin{array}{ccccccc}
 0       &         &          &                                                     \\
         & \ddots  &          &                                                     \\
         &         & 0        &                                                     \\
         &         &          &  \sigma_k   &                                       \\
         &         &          &             &  0                                    \\
         &         &          &             &           & \ddots      &              \\ 
         &         &          &             &           &             &  0          \\     
\end{array} \right]_{m \times n}  \longleftarrow (k \text{-th row})
$$ 

Finally, we prove that all matrices $\{ X_i \}_{i = 1, \cdots ,r }$ are diagonal. 
So the original matrix $X = \sum \limits_{i = 1}^{r} X_i $ is equal to the diagonal matrix $D$.
The other direction is obvious. 
We finish the proof.
\end{proof}




\begin{thebibliography}{1}

\bibitem{hard-threshold-blumensath2012accelerated}T. Blumensath, 
{\em Accelerated iterative hard thresholding}, Signal Processing, 92(3), pp. 752--756, 2012.


\bibitem{cai2010singular}
J. Cai, E. J. Cand{\`e}s, and Z. Shen.
\newblock A singular value thresholding algorithm for matrix completion.
\newblock {\em SIAM Journal on Optimization}, 20(4):1956--1982, 2010.

\bibitem{candes2010power}
E. J. Cand{\`e}s and T. Tao.
\newblock The power of convex relaxation: Near-optimal matrix completion.
\newblock {\em Information Theory, IEEE Transactions on}, 56(5):2053--2080, 2010.

\bibitem{CR_09} E. Cand\`es, and B. Recht, 
{\em Exact matrix completion via convex optimization}, Found. Comput. Math., 9 (2009), pp. 717-772.

\bibitem{xian-2-3rds-cao2013fast-image} W. Cao, J. Sun, and Z. Xu,
{\em Fast image deconvolution using closed-form thresholding 
formulas of $L_q, (q=1/2,2/3)$ regularization},
Journal of Visual Communication and Image Representation, 24(1), pp. 31--41, 2013.

\bibitem{chen2013low}
Y. Chen, A. Jalali, S. Sanghavi, and C. Caramanis.
\newblock Low-rank matrix recovery from errors and erasures.
\newblock {\em Information Theory, IEEE Transactions on}, 59(7):4324--4337,
  2013.

\bibitem{Daub_10}I. Daubechies, R. DeVore, M. Fornasier, C. Gunturk, 
{\em Iteratively reweighted least squares minimization for sparse recovery}, Comm. Pure Applied Math, 
63(1), pp. 1--38, 2010. 

\bibitem{soft-threshold-lp-daubechies2004iterative}
I. Daubechies, M. Defrise, and C. De~Mol.
\newblock An iterative thresholding algorithm for linear inverse problems with
  a sparsity constraint.
\newblock {\em Communications on pure and applied mathematics},
  57(11):1413-1457, 2004.
  
\bibitem{Don_95}D. Donoho, 
{\em Denoising by soft-thresholding}, IEEE Trans. Info. Theory, 41(3), pp. 613--627, 1995.

\bibitem{drineas2006fastI}
P. Drineas, R. Kannan, and M. W. Mahoney.
\newblock Fast monte carlo algorithms for matrices i: Approximating matrix
  multiplication.
\newblock {\em SIAM Journal on Computing}, 36(1):132--157, 2006.

\bibitem{drineas2006fastII}
P. Drineas, R. Kannan, and M. W. Mahoney.
\newblock Fast monte carlo algorithms for matrices ii: Computing a low-rank
  approximation to a matrix.
\newblock {\em SIAM Journal on Computing}, 36(1):158--183, 2006.

\bibitem{fan2001variable} J. Fan, and R. Li, 
{\em Variable selection via nonconcave penalized likelihood and its oracle properties,}
Journal of the American Statistical Association, 96(456):1348--1360, 2001.

\bibitem{KFan_inequality} K. Fan, 
{\em Maximum properties and inequalities for the eigenvalues of completely continuous operators,} 
Proc. Nat. Acad. Sci. U.S.A. 37 (1951), 760–766.

\bibitem{Fazel_01} M. Fazel, H. Hindi, and S. Boyd, 
{\em A rank minimization heuristic with application to minimum order system approximation}, 
 In Proc. American Control Conference, Arlington, VA, 2001.

\bibitem{Fazel_03}M. Fazel, H. Hindi, and S. Boyd, 
{\em Log-det heuristic for matrix rank minimization with applications to 
Hankel and Euclidean distance matrices}, 
in Proc. Amer. Control Confer., pp. 2156--2162, Denver, CO, 2003.

\bibitem{halko2011finding}
N. Halko, P. G. Martinsson, and J. A. Tropp.
\newblock Finding structure with randomness: Probabilistic algorithms for
  constructing approximate matrix decompositions.
\newblock {\em SIAM review}, 53(2):217--288, 2011.

\bibitem{Jannach_12}D. Jannach, M. Zanker, A. Felfernig, G. Friedrich, 
``Recommender Systems: An Introduction'', Cambridge Univ. Press, 2012. 

\bibitem{JiYe_13}S. Ji, K-F Sze, Z. Zhou, A. So, Y. Ye, 
{\em Beyond Convex Relaxation: A Polynomial-Time Non-Convex Optimization 
Approach to Network Localization}, Proceedings of the 32nd 
IEEE International Conference on Computer Communications (INFOCOM 2013), 2013, 
pp. 2499-2507. 

\bibitem{Kesh_10}R. Keshavan, A. Montanari, S. Oh,
{\em  Matrix completion from a few entries},
IEEE Trans. Info. Theory, 56 (6), 2980-2998, 2010.

\bibitem{IRucLq}
M. Lai, Y. Xu, and W. Yin.
\newblock Improved iteratively reweighted least squares for unconstrained
  smoothed $l_q$ minimization.
\newblock {\em SIAM Journal on Numerical Analysis}, 51(2):927--957, 2013.

\bibitem{lu2014iterative}
Z. Lu and Y. Zhang.
\newblock Iterative reweighted singular value minimization methods for $ l\_p $
  regularized unconstrained matrix minimization.
\newblock {\em arXiv preprint arXiv:1401.0869}, 2014.

\bibitem{transformed-l1}J. Lv, and Y. Fan, 
{\em A unified approach to model selection and sparse recovery using regularized least squares,}
Annals of Statistics, 37(6A), pp. 3498--3528, September 2009.

\bibitem{FPCA}
S. Ma, D. Goldfarb, and L. Chen.
\newblock Fixed point and bregman iterative methods for matrix rank
  minimization.
\newblock {\em Mathematical Programming}, 128(1-2):321--353, 2011.

\bibitem{Fazel}
K. Mohan and M. Fazel.
\newblock Iterative reweighted algorithms for matrix rank minimization.
\newblock {\em The Journal of Machine Learning Research}, 13(1):3441--3473,
  2012.

\bibitem{nie2012low}
F. Nie, H. Huang, and C. Ding.
\newblock Low-rank matrix recovery via efficient schatten p-norm minimization.
\newblock In {\em Twenty-Sixth AAAI Conference on Artificial Intelligence},
  2012.

\bibitem{Fazel-siamreview}
B. Recht, M. Fazel, and P. A. Parrilo.
\newblock Guaranteed minimum-rank solutions of linear matrix equations via
  nuclear norm minimization.
\newblock {\em SIAM review}, 52(3):471--501, 2010.

\bibitem{terence2012topics}
T. Tao.
\newblock {\em Topics in random matrix theory}, volume 132.
\newblock American Mathematical Soc., 2012.

\bibitem{LMaFit}
Z. Wen, W. Yin, and Y. Zhang.
\newblock Solving a low-rank factorization model for matrix completion by a
  nonlinear successive over-relaxation algorithm.
\newblock {\em Mathematical Programming Computation}, 4(4):333--361, 2012.

\bibitem{xian-half}Z. Xu, X. Chang, F. Xu, H. Zhang,
{\em $L_{1/2}$ regularization: an iterative thresholding method}, 
IEEE Transactions on Neural Networks and Learning Systems, 23, pp. 1013--1027, 2012.

\bibitem{Threshold-TL1}
S. Zhang and J. Xin.
{\em Minimization of transformed $ l\_1 $ penalty: Closed form representation and 
iterative thresholding algorithms.}
to appear in Comm. Math Sciences, 2016

\bibitem{DCATL1}
S. Zhang and J. Xin.
\newblock Minimization of transformed $l_1$ penalty: Theory, difference of
  convex function algorithm, and robust application in compressed sensing.
\newblock {\em arXiv preprint arXiv:1411.5735}, 2014.

\bibitem{friedlander2014gauge}
Michael~P Friedlander, Ives Macedo, and Ting~Kei Pong.
\newblock Gauge optimization and duality.
\newblock {\em SIAM Journal on Optimization}, 24(4):1999--2022, 2014.

\bibitem{friedlander2015low}
Michael~P Friedlander and Ives Macedo.
\newblock Low-rank spectral optimization.
\newblock {\em arXiv preprint arXiv:1508.00315}, 2015.

\bibitem{vandereycken2013low}
Bart Vandereycken.
\newblock Low-rank matrix completion by riemannian optimization.
\newblock {\em SIAM Journal on Optimization}, 23(2):1214--1236, 2013.

\end{thebibliography}

          %


\end{document}